\documentclass[11pt]{article}

\def\Span{\textup{span}}

\def\mS{\mathbb{\Sigma}}

\usepackage[absolute,overlay]{textpos}
\usepackage{nicefrac}
\usepackage{float}
\usepackage[dvips]{graphicx}
\usepackage{psfrag}
\usepackage{amsmath,amsthm,amsfonts,amssymb,amscd}
\usepackage[table]{xcolor}
\usepackage{mathrsfs}
\usepackage{upgreek}
\usepackage{stackrel}
\usepackage{tipa}
\usepackage{accents}
\usepackage[multiple]{footmisc}
\usepackage{multirow} 
\usepackage{booktabs}

\numberwithin{equation}{section}

\definecolor{SmartBlue}{RGB}{51, 51, 255}

\usepackage{enumitem}

\setlist[itemize]{topsep=0pt, partopsep=0pt, parsep=0pt, itemsep=0pt}

\usepackage{hyperref}

\hypersetup{filecolor=magenta}
\hypersetup{colorlinks=true}
\hypersetup{urlcolor=SmartBlue}
\hypersetup{linkcolor=SmartBlue}
\hypersetup{pdfpagemode=UseNone}
\hypersetup{citecolor=SmartBlue}

\usepackage[absolute,overlay]{textpos}

\usepackage{changepage}

\usepackage{float}

\usepackage{amsmath,amsthm,amsfonts,amssymb,amscd}
\usepackage{mathrsfs}
\usepackage{upgreek}
\usepackage{stackrel}
\usepackage{tipa}
\usepackage{accents}
\usepackage[multiple]{footmisc}

\usepackage{multirow} 
\usepackage{booktabs}

\usepackage{multicol}

\usepackage{xfrac} 

\newtheorem{theorem}{Theorem}[section]
\newtheorem{proposition}[theorem]{Proposition}
\newtheorem{lemma}[theorem]{Lemma}
\newtheorem{definition}[theorem]{Definition}

\newtheorem{remark}[theorem]{Remark}

\newtheorem*{orient*}{Orientation Condition}
\newtheorem*{junction*}{Junction Conditions}

\newtheorem*{nonexphor*}{Non-expanding horizons}
\newtheorem*{weakisolhor*}{Weakly isolated horizons}
\newtheorem*{isolhor*}{Isolated horizons}

\def\ovkil{\ov{\teta}}
\def\ovkil{\ov{\teta}}
\def\teta{\eta}
\def\kil{\teta}

\def\defi{{\stackrel{\mbox{\tiny {\textbf{def}}}}{\,\, = \,\, }}}

\def\bsff{\textbf{\textup{K}}}
\def\btsff{\bsff{}}

\def\nablao{{\stackrel{\circ}{\nabla}}}

\def\Riemo{{\stackrel{\circ}{R}}}

\def\sone{\hat{s}}
\def\bsone{\bs{\hat{s}}}

\def\metdata{\{\mathcal{N},\gamma,\ellc,\elltwo\}}

\def\G{\mathfrak{G}}
\def\Hemb{\hor}
\def\H{\mathcal H}

\def\bomega{\bs{\omega}}

\def\bkilone{\bs{\upvarpi}}

\def\trP{\textup{tr}_P}
\def\n{\mathfrak{n}}

\def\G{\mathcal G}
\def\Fcal{\mathcal F}

\def\bY{\textup{\textbf{Y}}}

\def\Y{\textup{Y}}
\def\bU{\textup{\textbf{U}}}
\def\U{\textup{U}}

\def\bF{\textup{\textbf{F}}}
\def\F{\textup{F}}

\def\Yn{r}
\def\Q{Q}

\def\sigo{\stackrel{\circ}{\Sigma}}

\newcommand{\nn}{\nonumber}

\newcommand{\bm}[1]{\mbox{\boldmath $#1$}}
\newcommand\ovnabla{\ov{\nabla}}

\newcommand\N{\mathcal H}
\newcommand\M{\mathcal M}

\newcommand\elltwo{\ell^{(2)}}

\newcommand\hypdata{\{ \mathcal{H},\gamma,\ellc, \elltwo, \bY\}}
\newcommand\rig{\xi}

\newcommand\A{\mathcal A}

\def\bmell{\bm{\ell}}
\def\ellc{\bmell}

\def\nablao{{\stackrel{\circ}{\nabla}}}
\def\K{\mathcal K}
\def\Kkil{\K[\kil]}
\def\SigmaZ{\mS[Z]}
\def\Sigmakil{\Sigma[\kil]}

\newcommand{\Rtensor}{\mathcal{R}}

\def\defi{{\stackrel{\mbox{\tiny \textup{\textbf{def}}}}{\,\, = \,\, }}}

\def\nablao{{\stackrel{\circ}{\nabla}}}

\def\Riemo{{\stackrel{\circ}{R}}}

\def\sone{s}
\def\bsone{\bs{s}}

\def\metdata{\{\mathcal{H},\gamma,\ellc,\elltwo\}}

\def\trP{\textup{tr}_P}
\def\n{\mathfrak{n}}

\def\G{\mathcal G}
\def\Fcal{\mathcal F}

\def\bY{\textup{\textbf{Y}}}

\def\Y{\textup{Y}}

\def\bU{\textup{\textbf{U}}}
\def\U{\textup{U}}

\def\bF{\textup{\textbf{F}}}
\def\F{\textup{F}}

\def\Yn{r}
\def\Q{\kappa_n}

\def\bmell{\bm{\ell}}
\def\ellc{\bmell}

\def\nablao{{\stackrel{\circ}{\nabla}}}
\def\K{\mathcal K}

\newcommand\Rad{\mbox{Rad}}

\usepackage{xspace}
\newcommand{\textbothdata}{(metric) $\mathcal{K}$-tuple\xspace}
\newcommand{\textnonedata}{$\mathcal{K}$-tuple\xspace}
\newcommand{\textpmdata}{metric $\mathcal{K}$-tuple\xspace}
\newcommand{\pmdata}{\{\N,\gamma,\ellc,\elltwo,\alpha,\p,\bqone\}}
\newcommand{\textphdata}{$\mathcal{K}$-tuple\xspace}
\newcommand{\phdata}{\{\N,\gamma,\ellc,\elltwo,\bY,\alpha,\p,\bqone\}}

\newcommand{\isotensor}{non-isolation tensor\xspace}

\newcommand{\ov}{\overline}

\newcommand{\bs}{\boldsymbol}
\newcommand{\lp}{\left(}
\newcommand{\rp}{\right)}

\newcommand{\cv}{\mathcal{V}}
\newcommand{\lieo}{\mathsterling}

\def\metdataa{\{\gamma,\ellc,\elltwo\}}

\newcommand{\wt}{\widetilde}

\newcommand{\nullhyp}{\widetilde{\N}}

\newcommand{\spc}{\textup{ }}

\newcommand{\ovsigmakil}{\ov{\Sigma}[\ovkil]}

\usepackage[new]{old-arrows}

\newcommand{\w}{\mathfrak{w}}
\newcommand{\p}{\mathfrak{p}}
\newcommand{\qone}{\mathfrak{q}}
\newcommand{\bqone}{\bs{\qone}}

\def\calP{\Pi}
\def\bcalP{\bs{\Pi}}

\makeatletter
\newsavebox\myboxA
\newsavebox\myboxB
\newlength\mylenA

\newcommand*\xoverline[2][0.75]{%
    \sbox{\myboxA}{$\m@th#2$}%
    \setbox\myboxB\null
    \ht\myboxB=\ht\myboxA%
    \dp\myboxB=\dp\myboxA%
    \wd\myboxB=#1\wd\myboxA
    \sbox\myboxB{$\m@th\overline{\copy\myboxB}$}
    \setlength\mylenA{\the\wd\myboxA}
    \addtolength\mylenA{-\the\wd\myboxB}%
    \ifdim\wd\myboxB<\wd\myboxA%
       \rlap{\hskip 0.5\mylenA\usebox\myboxB}{\usebox\myboxA}%
    \else
        \hskip -0.5\mylenA\rlap{\usebox\myboxA}{\hskip 0.5\mylenA\usebox\myboxB}%
    \fi}
\makeatother

\DeclareFontFamily{U}{rcjhbltx}{}
\DeclareFontShape{U}{rcjhbltx}{m}{n}{<->rcjhbltx}{}
\DeclareSymbolFont{hebrewletters}{U}{rcjhbltx}{m}{n}

\DeclareMathSymbol{\aleph}{\mathord}{hebrewletters}{39}
\DeclareMathSymbol{\beth}{\mathord}{hebrewletters}{98}
\DeclareMathSymbol{\gimel}{\mathord}{hebrewletters}{103}
\DeclareMathSymbol{\lamed}{\mathord}{hebrewletters}{108}
\DeclareMathSymbol{\mem}{\mathord}{hebrewletters}{109}
\DeclareMathSymbol{\ayin}{\mathord}{hebrewletters}{96}
\DeclareMathSymbol{\tsadi}{\mathord}{hebrewletters}{118}
\DeclareMathSymbol{\qof}{\mathord}{hebrewletters}{113}
\DeclareMathSymbol{\shin}{\mathord}{hebrewletters}{152}
\DeclareMathSymbol{\pe}{\mathord}{hebrewletters}{80}
\DeclareMathSymbol{\heh}{\mathord}{hebrewletters}{104}
\DeclareMathSymbol{\peh}{\mathord}{hebrewletters}{112}
\newcommand{\hor}{\mathscr{H}}

 \newcounter{mnotecount}

 \newcommand{\mnote}[1]
 {\protect{\stepcounter{mnotecount}}$^{\mbox{\tiny
 $\,\bullet$\themnotecount}}$ \marginpar{
 \raggedright\tiny\em
 $\,\bullet$\themnotecount: #1} }

\allowdisplaybreaks

\usepackage[a4paper,top=30mm,bottom=30mm,
left=20mm, right=20mm,bindingoffset=0mm]{geometry} 

\title{Horizon Data: Existence Results and a\\ Near-Horizon Equation on General Null Hypersurfaces}

\author{
	Miguel Manzano$^1$\thanks{{\tt m.manzano.rod@gmail.com}}\ \ and
	Marc Mars$^2$\thanks{{\tt marc@usal.es}}\\ \\
	$^1$ Faculty of Mathematics, University of Vienna, \\
	Oskar-Morgenstern-Platz 1, 1090 Vienna, Austria. \\ \\
    $^2$ Faculty of Sciences, University of Salamanca,\\
    Plaza de la Merced 18, 37008 Salamanca, España.
}

\makeatletter
\newcommand\subsubsubsection{\@startsection{paragraph}{4}{\z@}{-2.5ex\@plus -1ex \@minus -.25ex}{1.25ex \@plus .25ex}{\normalfont\normalsize\bfseries}}
\newcommand\subsubsubsubsection{\@startsection{subparagraph}{5}{\z@}{-2.5ex\@plus -1ex \@minus -.25ex}{1.25ex \@plus .25ex}{\normalfont\normalsize\bfseries}}
\makeatother

\usepackage{tabularx}

\usepackage{tocloft}

\usepackage[normalem]{ulem}

\usepackage{bbold}

\AfterEndEnvironment{theorem}{\vspace{-0.5em}}
\AfterEndEnvironment{proposition}{\vspace{-0.5em}}
\AfterEndEnvironment{lemma}{\vspace{-0.5em}}
\AfterEndEnvironment{proof}{\vspace{-0.5em}}
\AfterEndEnvironment{remark}{\vspace{-0.5em}}

\begin{document}

\setlength{\abovedisplayskip}{0.15cm}
\setlength{\belowdisplayskip}{0.15cm}

\maketitle

\begin{abstract}
In a spacetime $(\mathcal{M},g)$, a horizon is a null hypersurface where the deformation tensor $\mathcal{K}:=\pounds_{\eta}g$ of a null and tangent vector $\eta$ satisfies certain restrictions. In this work, we develop a formalism to study the geometry of \textit{general} horizons (i.e.\ characterized by any $\mathcal{K}$), based on encoding the zeroth and first transverse derivatives of $\mathcal{K}$ on null hypersurfaces detached from any ambient spacetime. We introduce the notions of \textit{$\mathcal{K}$-tuple} and \textit{non-isolation tensor}. The former encodes the order zero of $\mathcal{K}$, while the latter is a symmetric $2$-covariant tensor that codifies the ``degree of isolation" of a horizon. In particular, the non-isolation tensor vanishes for homothetic, Killing and isolated horizons. As an application we derive a \textit{generalized near-horizon equation}, i.e., an identity that holds on any horizon (regardless of its topology or whether it contains fixed points), which relates the non-isolation tensor, a certain torsion one-form, and curvature terms. By restricting this equation to a cross-section one can recover the near-horizon equation of isolated horizons and the master equation of multiple Killing horizons. Our formalism allows us to prove two existence theorems for horizons. Specifically, we establish the necessary and sufficient conditions for a non-degenerate totally geodesic horizon with any prescribed non-isolation tensor to be embeddable in a spacetime satisfying any (non-necessarily $\Lambda$-vacuum) field equations. We treat first the case of arbitrary topology, and then show how the result can be strengthened when the horizon admits a cross-section. 
\end{abstract}

\section{Introduction}\label{c:Introduction}

For the purposes of this paper, horizons are null hypersurfaces endowed with a privileged null and tangent vector field $\kil$ which encodes certain symmetry property of the spacetime. They play a fundamental role in many gravitational contexts such as gravitational collapse, black hole uniqueness, singularity theorems, or null infinity, see e.g.\ 
\cite{
  penrose1969gravitational,
  wald1984general,
  senovilla1998singularity,
  chrusciel2012stationary,
  frolov2012black,
  kunduri2013classification,
  ashtekar2024null,
  hounnonkpe2025horizon} 
and references therein. Prime examples of horizons are homothetic, conformal Killing and Killing horizons, but there are many other useful notions such as non-expanding, weakly isolated and isolated horizons \cite{ashtekar2000generic,ashtekar2000isolated, ashtekar2002geometry,krishnan2002isolated,gourgoulhon20063+, jaramillo2009isolated}, 
or multiple Killing horizons \cite{mars2018multiple}. Horizons usually mark the boundary of causally relevant regions of a spacetime, and their presence reveals interesting connections between its causal structure and the behaviour of symmetries, which often change causal character across the horizons. In particular, Killing horizons are linked to (local) isometries and hence to equilibrium states, whereas homothetic horizons arise for instance in self-similar evolution, most notably in critical gravitational collapse \cite{gundlach2007critical,rodnianski2018asymptotically}. In a more general context, homothetic horizons are also prominent objects in conformal geometry via the Fefferman-Graham ambient construction \cite{fefferman1985conformal,fefferman2012ambient}.
\medskip
 
The link between horizons and spacetime symmetries raises fundamental questions, such as what are the conditions for a spacetime to admit a horizon, or what data must be prescribed on a horizon so that it induces a specific type of symmetry. These issues have been widely studied in the literature, starting with the seminal work \cite{moncrief1982neighborhoods} where analytic, Ricci-flat ($4$-dimensional) spacetimes with a non-degenerate Killing horizon were shown to be uniquely determined by six functions on the horizon (a result that was later reformulated in more geometric terms, and extended to smooth spacetimes in \cite{kroencke2024asymptotic}). For \textit{degenerate} Killing horizons, several non-existence and uniqueness results have been proved, 
including non-existence of such horizons in static $\Lambda$-vacuum spacetimes \cite{chrusciel2005non,bahuaud2022static,wylie2023rigidity}, uniqueness of Schwarzchild-de Sitter \cite{katona2024uniqueness} and (local) uniqueness of the Kerr-horizon data in the (electro)vacuum rotating case 
\cite{lewandowski2003extremal,kunduri2009classification,kunduri2009uniqueness,matejov2021uniqueness,dunajski2025intrinsic}. 
In the case where the horizon is of bifurcate type, i.e.\ it consists of two transverse null hypersurfaces, the problem can be analyzed from a characteristic initial value perspective. Relevant results in this regard are the formulation of the characteristic problem on bifurcate Killing horizons \cite{racz2007stationary} (see also \cite{chrusciel2012manyways}), the analyticity of the spacetime metric for bifurcate Killing horizons defined by a hypersurface-orthogonal Killing vector \cite{chrusciel2004analyticity}, and the fact that bifurcate Killing horizons with closed torsion one-form give rise to (and arise from) static Killing vectors \cite{mars2023staticity}. Motivated by these works, the purpose of the present paper is to introduce a formalism for describing horizons as manifolds a priori \textit{detached from any ambient spacetime}, and to exploit it to study the interplay between the existence of  spacetimes with symmetries and the geometric properties of the horizons. The key point of the formalism is that it allows one to work with \textit{generic} notions of horizons (i.e.\ far less restrictive than those of homothetic or a Killing horizons),  
as well as to deal with spacetimes satisfying \textit{any} (i.e.\ non-necessarily $\Lambda$-vacuum) field equations. 
\medskip  

A particularly relevant object in this context is the so-called deformation tensor $\Kkil\defi \pounds_{\kil}g$ of a spacetime $(\M,g)$, which accounts for the ``degree of symmetry" of $(\M,g)$ along $\kil$. In fact, every notion of horizon imposes constraints on $\Kkil$ and/or its derivatives. For instance, \(\Kkil\) and all its derivatives vanish on Killing horizons, 
and are proportional to derivatives of \(g\) on homothetic horizons, whereas isolated horizons impose restrictions only up to the first transverse derivative of \(\Kkil\) \cite{manzano2023field}. We base our approach on capturing the geometric information of $\Kkil$ (more specifically, its zeroth and first orders at the horizon) in terms of \textit{data}---i.e.\ a set of tensor fields---on a detached null hypersurface. In this way, the formalism enables us to consider \textit{completely general} horizons (i.e.\ with arbitrary $\Kkil$ and with no restrictions of the topology of the horizon), and then study specific types of horizon by particularizing the data set. Besides, we shall maintain full generality regarding the generator $\kil$, which in particular is allowed to vanish on the horizon.
\medskip

To capture the order zero of $\Kkil$ on a 
detached null hypersurface $\N$ we introduce the concept of \textit{\textpmdata} (cf.\ Definitions \ref{def:phdata} and \ref{def:embedded:pdata}), which involves two functions $\alpha,\p$ and a covector $\bqone$, subject to certain restrictions. This notion is tailored so that, once $\N$ is embedded in a spacetime, the deformation tensor at the horizon is given by the quantities $\alpha \bU, \p,\bqone$, where $\bU$ is computable from the data and a posteriori agrees with the second fundamental form of $\N$. In particular, a Killing horizon of order zero is described at a non-embedded level by a \textphdata with $\{\alpha\bU=0,\p=0,\bqone=0\}$. A key result in this context is that, for arbitrary $\{\alpha,\p,\bqone\}$, any embedded null hypersurface can be turned into a zeroth-order horizon with deformation tensor given by $\alpha\bU,\p,\bqone$ (Proposition \ref{prop:extension:eta}). This is achieved by proving that the vector $\kil$ can always be extended off the horizon in such a way that its deformation tensor matches the prescribed data. The extension of $\kil$ is of course not unique, but its first order transverse derivative at the horizon is uniquely determined (equation \eqref{lie:rig:eta}).
\medskip

The task of capturing the order one of $\Kkil$ (i.e.\ its first transverse derivative) at a non-embedded level is more complicated, and involves commutators of Lie and covariant derivatives.  For instance, isolated horizons \cite{ashtekar2002geometry} are totally geodesic null hypersurfaces for which the commutator $[\pounds_{\kil},\ovnabla]$ of the connection $\ovnabla$ induced from the spacetime and the Lie derivative along $\kil$ vanishes. For totally geodesic null hypersurfaces, $\ovnabla$ is intrinsic to the horizon, hence $[\pounds_{\kil},\ovnabla]=0$ is a property of the horizon and one can understand the tensor
$[\pounds_{\kil},\ovnabla]$ as a measure of the ``degree of non-isolation". For general null hypersurfaces, however, the induced connection $\ovnabla$ depends strongly on the choice of a \textit{rigging} (i.e.\ a transverse vector along the hypersurface). Consequently,  $[\pounds_{\kil},\ovnabla]$ is in general not suitable for characterizing horizon geometry, and one must find a tensor which encodes the same information as $[\pounds_{\kil},\ovnabla]$ in the totally geodesic case, but exhibits a better behavior under changes of rigging. In Section \ref{sec:hor:data}, we provide one such notion of \isotensor $\bcalP^{\kil}$ (Definition \ref{Def:calP:sigma:new}), and show that in the embedded case it is related to the order one of $\Kkil$ (Lemma \ref{lem:order_one}). In particular, $\bcalP^{\kil}$ vanishes for isolated, homothetic and Killing horizons.
\medskip

An interesting application of the results above is the so-called \textit{generalized near-horizon equation} (Theorem \ref{thm:master:equation}). This is an identity that holds on \textit{any}  \textphdata, which relates the \isotensor, the surface gravity $\kappa$ of $\kil$, the extrinsic curvature of $\N$, a one-form $\bomega$, and the so-called \textit{constraint tensor} $\bs{\Rtensor}$ \cite{mars2023first,manzano2023constraint} (namely a tensor that encodes the tangential part of the ambient Ricci tensor at a non-embedded level). 
The one-form $\bomega$ is an extension of the well-known \textit{torsion one-form} in the sense that $\bomega$ restricted to cross-sections (when they exist) coincides with it. Remarkably, the generalized near-horizon equation is valid for arbitrary topology (in particular, it does not require the existence of a cross-section), and for \textnonedata{}s possibly containing fixed points of $\kil$. It also holds for arbitrary $\bs{\Rtensor}$, hence for any field equations one may wish to impose. More important is the fact that its restriction to cross-sections enables to recover both the \textit{near-horizon equation} of isolated horizons \cite[Eq.\ (5.3)]{ashtekar2002geometry} and the \textit{master equation} of multiple Killing horizons \cite[Eq.\ (60)]{mars2018multiple} as particular cases. For the purposes of this work, one key aspect of the generalized near-horizon equation is that, when $\kappa$ is no-where zero, it allows one to determine algebraically the extrinsic curvature of the horizon (i.e.\ first transverse derivatives of the spacetime metric in the embedded case) from $\bomega$, $\bcalP^{\kil}$ and $\bs{\Rtensor}$. This result turns out to be essential to prove existence of a spacetime with a symmetry from initial data prescribed on a horizon (Section \ref{sec:existence:results}).
\medskip

Two main results of the paper arise from applying the framework to study the problem of how to construct a spacetime from data on a totally geodesic horizon with nowhere zero $\kappa$. We obtain two existence results (Theorems \ref{thm:exist:pi:R} and \ref{thm:existence:sections}). The first one establishes the necessary and sufficient conditions to be satisfied by the data to ensure existence of a spacetime $(\M,g)$ realizing the prescribed quantities. More precisely, the data consists of the order zero of the spacetime metric together with freely prescribed quantities $\alpha$, $\bomega$, $\bcalP^{\kil}$ and $\bs{\Rtensor}$ subject only to certain differential conditions, and the spacetime $(\M,g)$ is such that $\alpha$ determines the vector field $\kil$ on the embedded horizon, $\bomega$ is the generalized torsion one-form, $\bcalP^{\kil}$ is the \isotensor of $\kil$, and $\bs{\Rtensor}$ coincides with the tangential components of the ambient Ricci tensor. The necessary and sufficient conditions, on the other hand, are 
a set of differential equations that, in particular, provide evolution equations for $\alpha$, $\bomega$ and $\bcalP^{\kil}$ (or $\bs{\Rtensor}$) along the null generators of the horizon (cf.\ \eqref{Pi(n,dot)=Rtensor}-\eqref{Lie:Pi}). 
These results therefore determine how to build spacetimes containing totally geodesic horizons with any ``degree of non-isolation", any topology, possibly with fixed points, and such that the spacetime Ricci tensor satisfies any field equations. Since horizons admitting a cross-section $S$ are of particular relevance in General Relativity, in Theorem \ref{thm:existence:sections} we provide the corresponding existence theorem with initial data on $S$. As a particular case, Theorems \ref{thm:exist:pi:R} and \ref{thm:existence:sections} can be applied to construct spacetimes with non-expanding, weakly isolated or Killing horizons.\medskip

In the particular case when the non-isolation tensor $\bcalP^{\kil}$ vanishes identically and the constraint tensor $\bs{\Rtensor}$ is $\Lambda$-vacuum, the necessary and sufficient conditions \eqref{Pi(n,dot)=Rtensor}-\eqref{Lie:Pi} simplify drastically. The third condition becomes trivial, while the first one just states that the surface gravity is constant. Thus, the only relevant condition that remains is \eqref{Pi(n,dot)=Lie:omega}. Up to a simple redefinition of variables, this is precisely the defining condition of the notion of abstract Killing horizon data (AKH) introduced in \cite{mars2024transverseII}. In \cite{mars2024transverseII} and its follow up paper \cite{mars2025KID}, AKH data with constant non-zero surface gravity was shown to give rise to a spacetime   $(\M,g)$ where the data can be embedded as a Killing horizon, and such that the $\Lambda$-vacuum field equations are satisfied to infinite order on the horizon. Asymptotic uniqueness of $(\M,g)$ at the horizon was also proven.  Given that in the present paper we are not assuming any field equations, the existence theorem we establish is necessarily weaker concerning the properties of the spacetime one constructs (nothing can be said about higher order transverse derivatives at the horizon). However, it is stronger in the sense that it allows for general non-isolation tensor $\bcalP^{\kil}$ and constraint tensor $\bs{\Rtensor}$. In combination with suitable field equations, our result can potentially give rise to existence and uniqueness theorems to all orders for spacetimes beyond $\Lambda$-vacuum and/or admitting more general types of horizons.\medskip

The structure of the paper is as follows. In Section \ref{sec:prelim} we revisit the basic concepts and results of the formalism of hypersurface data 
\cite{mars2013constraint,mars2020hypersurface,mars1993geometry,manzano2023matching,manzano2023constraint}, which is the mathematical framework used in the paper. Section \ref{sec:tensors:sigma} is devoted to studying commutators of Lie and covariant derivatives.  
In Section \ref{sec:hor:data} we introduce the notions of \textphdata and \isotensor, which are then exploited in Section \ref{sec:Cov:ME:General:Hyp} to derive the generalized near-horizon equation. Finally, in Section \ref{sec:existence:results} we demonstrate the aforementioned existence results in the case of totally geodesic initial data. 

\subsection{Notation and conventions}\label{sec:notation}

In this paper, all manifolds are smooth, connected and without boundary. Given a manifold $\M$ we use $\mathcal{F}\lp\mathcal{M}\rp\defi C^{\infty}\lp\mathcal{M},\mathbb{R}\rp$, and $\mathcal{F}^{\star}\lp\mathcal{M}\rp\subset\mathcal{F}\lp\mathcal{M}\rp$ for the subset of no-where zero functions. The tangent bundle is denoted 
by $T\mathcal{M}$, and $\Gamma\lp T\mathcal{M}\rp$ is the set
of sections (i.e.\ vector fields). We use $\pounds$, $d$ for the Lie derivative and exterior derivative. Both tensorial and abstract index notation will be employed. We work in arbitrary dimension $\mathfrak{n}\geq 2$ and use the following sets of indices:
\begin{equation}
\label{notation}
\alpha,\beta,...=0,1,2,...,\mathfrak{n};\qquad a,b,...=1,2,...,\mathfrak{n};\qquad A,B,...=2,...,\mathfrak{n}.
\end{equation}
When index-free notation is used (and only then) we shall distinguish
covariant tensors with boldface.\ As usual, parenthesis (resp.\ brackets) denote symmetrization (resp.\ antisymmetrization) of indices.\ The symmetrized tensor product is defined by $A\otimes_s B\equiv\frac{1}{2}(A\otimes B+B\otimes A)$.\  
We write $\text{tr}_B\bs{A}$ for the trace of a symmetric $(0,2)$-tensor $\bs{A}$ with respect to a $(2,0)$-tensor $B$. Our notation and convention for the curvature operator of a connection $D$ is 
\begin{align}
\label{curvoperator} R^{D}(X,W) Z \defi
\left ( D_X D_W - D_W D_X - D_{[X,W]} \right ) Z, 
\end{align}
and we write $\textbf{\textup{Hess}}^D$ for the Hessian operator of $D$. In any Lorentzian manifold $(\mathcal{M},g)$,
the scalar product of two vectors is written as $g(X,Y)$, and we use $g^{\sharp}$ and $\nabla$ for the inverse and the Levi-Civita derivative of $g$ respectively. Our signature for Lorentzian manifolds $\lp \mathcal{M},g\rp$ is $(-,+, ... ,+)$.

\section{Preliminaries: The formalism of hypersurface data}\label{sec:prelim}
In this section, we present the basic notions of the hypersurface data formalism, and collect a number of results for later use. Since this work deals with null hypersurfaces, we only present the formalism in the null case. For further details we refer to \cite{mars2013constraint,mars2020hypersurface} (see also \cite{mars1993geometry,manzano2023matching,manzano2023constraint}).

Let $\N$ be a smooth $\n$-manifold endowed with a symmetric $(0,2)$-tensor field $ \gamma $ of signature $(0,+,...,+)$, a covector field $\ellc$ and a scalar function $\ell^{(2)}$. The tuple $\metdata$ defines \textbf{null metric hypersurface data} provided that the square $(\n+1)$-matrix
\begin{equation}
\label{def:A}\bs{\A}\defi \lp \hspace{-0.1cm}
\begin{array}{cc}
\gamma_{ab} & \ell_a\\
\ell_b & \elltwo
\end{array}
\hspace{-0.1cm}\rp
\end{equation}
is non-degenerate everywhere on $\N$.\  
A null metric data $\metdata$ equipped with an additional symmetric $(0,2)$-tensor $\bY$ defines \textbf{null hypersurface data} $\hypdata$. 

Given null metric hypersurface data, one can uniquely define a symmetric $(2,0)$-tensor field $P$ and a vector field $n$ as the entries of the inverse of $\bs{\A}$ \cite{mars2013constraint}. By construction, they satisfy

\vspace{-0.55cm}

\noindent
\begin{minipage}[t]{0.18\textwidth}
	\begin{align}
		\gamma_{ab} n^b & = 0, \label{prod1} 
	\end{align}
\end{minipage}
\hfill
\hfill
\begin{minipage}[t]{0.18\textwidth}
	\begin{align}
		\ell_a n^a & = 1, \label{prod2}  
	\end{align}
\end{minipage}
\hfill
\begin{minipage}[t]{0.3\textwidth}
	\begin{align}
		P^{ab} \ell_b + \elltwo n^a & = 0,  \label{prod3} 
	\end{align}
\end{minipage}
\hfill
\begin{minipage}[t]{0.3\textwidth}
	\begin{align}
		P^{ab} \gamma_{bc} + n^a \ell_c & = \delta^a_c. \label{prod4}
	\end{align}
\end{minipage}

\vspace{0.1cm}

Note that $n$ is no-where zero and spans the radical $\text{Rad}\gamma\vert_p\defi\{X\in T_p\N\hspace{0.05cm}\vert\hspace{0.05cm} \gamma(X,\cdot)=0\}$ of the degenerate tensor $\gamma$ (cf.\ \eqref{prod1}-\eqref{prod2}). We call \textbf{generators} of $\N$ the integral curves of $n$. 

\begin{remark}
We have fixed the signature of $\gamma$ to be $(0,+,\dots,+)$  for concreteness. All results below apply when $\gamma$ has signature $(0,-,\dots,-,+,\dots,+)$. The only difference is that the ambient spaces $(\M,g)$ are no longer Lorentzian but  semi-Riemannian (with same signature as the tensor $\bs{\mathcal{A}}$ in \eqref{def:A}).
\end{remark}

The following tensor fields play an important role in the hypersurface data formalism \cite{mars2020hypersurface}:

\vspace{-0.6cm}

\noindent
\begin{minipage}[t]{0.48\textwidth}
\begin{align}
\label{threetensors} \bF & \defi \frac{1}{2} d \ellc, & \bm{\sone} &\defi  \bF(n,\cdot), & \bU  &\defi  \frac{1}{2}\pounds_{n} \gamma ,
\end{align}
\end{minipage}
\hfill
\hfill
\begin{minipage}[t]{0.52\textwidth}
\begin{align}
\label{defY(n,.)andQ}
\bs{\Yn}&\defi\bY(n,\cdot), & \bomega&\defi\bsone-\bs{\Yn}, & \Q &\defi -\bY(n,n).
\end{align}
\end{minipage}

\vspace{-0.cm}

Observe that $\bU$ is symmetric and $\bF$ is a $2$-form, hence $\bsone(n)=0$.\ Moreover, $\bsone$ and $\bU$ verify 
\cite{mars2020hypersurface}

\vspace{-0.6cm}

\noindent
\begin{minipage}[t]{0.4\textwidth}
	\begin{align}
		\pounds_{n} \ellc & = 2 \bm{\sone}, \label{soneprop}
	\end{align}
\end{minipage}
\hfill
\hfill
\begin{minipage}[t]{0.6\textwidth}
	\begin{align}
		\bU (n, \cdot ) & = 0. \label{Un} 
	\end{align}
\end{minipage}

Since $\N$ is not endowed with a metric tensor, we cannot define a Levi-Civita covariant derivative. However, it turns out that 
there exists a canonical notion of covariant derivative on $\N$. Specifically, given null metric hypersurface data $\metdata$, the two conditions \cite[Prop. 4.3]{mars2020hypersurface}

\vspace{-0.55cm}

\noindent
\begin{minipage}[t]{0.4\textwidth}
	\begin{align}
		\nablao_{a} \gamma_{bc} & = - \ell_b \U_{ac} - \ell_c \U_{ab}, \label{nablaogamma} 
	\end{align}
\end{minipage}%
\hfill
\hfill
\begin{minipage}[t]{0.6\textwidth}
	\begin{align}
		\nablao_a \ell_b & = \F_{ab} - \elltwo \U_{ab} \label{nablaoll},
	\end{align}
\end{minipage}

uniquely define a torsion-free covariant derivative $\nablao$ on $\N$, called
\textbf{metric hypersurface connection}. 
The curvature and Ricci tensors of $\nablao$ are denoted by $\Riemo{}^{a}{}_{bcd}$ and $\Riemo_{ab}$ respectively.  
The tensor $\Riemo_{ab}$ verifies $\Riemo_{[ab]} = \nablao_{[a} \sone_{b]}$ \cite{mars2020hypersurface}, so it is not symmetric in general. The following three identities involving $\nablao$ and $\Riemo_{ab}$ will be used later \cite{mars2020hypersurface,manzano2023constraint}:
\begin{align}
\label{nablao:n}\nablao_{a}n^b=&\spc n^b\sone_a+P^{bf}\U_{af},\\
n^b\nablao_{(a}\theta_{b)}=&\spc\frac{1}{2}\pounds_{n}\theta_a  +  \frac{1}{2} \nablao_{a}(\bs{\theta}({n}))  - \bs{\theta}(n) \sone_a    -   P^{bc}\theta_{b}\U_{ac},\quad \forall \bs{\theta}\in\Gamma(T^{\star}\N),
\label{n:nablao:theta:sym}\\
\label{Riemosym(n,-)} \Riemo_{(ab)}n^b& =
    \dfrac{1}{2}\pounds_{n}\sone_a-2P^{cd}\U_{ac}\sone_{d}+P^{cd}\nablao_c\U_{ad} -\nablao_a(\textup{tr}_P\bU)+ (\textup{tr}_P\bU)\sone_a.
\end{align}
The connection $\nablao$ is not the only useful covariant derivative that can be constructed from a given data set. From hypersurface data one can  define another torsion-free connection $\ovnabla$ by \cite{mars2013constraint}
\begin{align}
\label{def:ovnabla} \ovnabla_XW\defi \nablao_XW-\bY(X,W)n,\qquad \forall X,W\in\Gamma(T\N).
\end{align}
We call $\ovnabla$ \textbf{hypersurface} 
\textbf{connection} since, as opposed to $\nablao$, it depends on the tensor field $\bY$.\ 

The notions of rigging and embedded data \cite{mars2013constraint,mars2020hypersurface} connect the hypersurface data formalism with the geometry of embedded hypersurfaces. 
A null metric data is $\bs{(\phi,\rig)}$\textbf{-embedded} in a Lorentzian $(\n+1)$-manifold $(\M,g)$ if there exists an embedding $\phi:\mathcal{N}\longhookrightarrow\mathcal{M}$ and a \textbf{rigging} $\rig$---i.e.\ a vector field along $\phi \lp\mathcal{N}\rp$, everywhere transverse to it---satisfying  
\begin{equation}
	\label{emhd}
	\phi ^{\star}\lp g\rp= \gamma , \qquad\phi ^{\star}\lp g\lp\rig,\cdot\rp\rp=\ellc, \qquad\phi ^{\star}\lp g\lp\rig,\rig\rp\rp=\elltwo.
\end{equation}
Embeddedness of hypersurface data $\hypdata$ requires the additional condition
\begin{equation}
\label{YtensorEmbDef}
\dfrac{1}{2}\phi^{\star}\lp \pounds_{\rig}g\rp=\bY.
\end{equation}
When no ambiguity arises, we shall identify scalars and vectors on $\N$ with their counterparts on $\phi(\N)$. We will also use the word ``abstract" for mathematical objects that can be defined irrespectively of whether or not the data is embedded in an ambient space, i.e., for quantities defined solely in terms of hypersurface data. Examples of abstract objects are $\gamma$, $\Q$ or the manifold $\N$. 

In the embedded case, 
the tensors $\gamma$ and $\bU$ are the first and second fundamental forms of $\phi(\N)$ respectively, while $\nu\defi \phi_{\star}n$ is the only vector field normal to $\phi(\N)$ and satisfying $g(\rig,\nu)=1$ \cite{mars2020hypersurface}. Moreover, the Levi-Civita derivative $\nabla$ of $g$ is related to the connections $\nablao$, $\ovnabla$ by \cite{mars2013constraint,mars2020hypersurface}
\begin{align}
\nabla_{X}W&=\ovnabla_{X}W - \bU(X,W)\rig
=\nablao_{X}W-\bY(X,W)\nu - \bU(X,W)\rig, \qquad\forall X,W\in\Gamma(T\N).
\label{nablaXYnablao}
\end{align}
Thus, $\ovnabla$ is the covariant derivative induced from the Levi-Civita connection of the ambient space when the rigging $\rig$ is used to split $T_{p}M$, $p\in\phi(\N)$ as a direct sum \cite{mars1993geometry}, i.e.\ $T_{p} M = T_p \phi(\H) \oplus \Span(\xi\vert_p)$.

The hypersurface data formalism encodes information about the ambient curvature without the need of considering the hypersurface as embedded.  
Specifically, from (non-embedded) null hypersurface data $\hypdata$ one can define a symmetric $(0,2)$-tensor $\bs{\Rtensor}$, called \textbf{constraint tensor}, as	\cite{mars2023first,mars2023zcovariant,manzano2023constraint}
\begin{align}
\nn \Rtensor_{ab} \defi  &\spc \Riemo_{(ab)}- 2 \pounds_{{n}} \Y_{ab}
- \lp 2\bomega(n)+\trP\bU \rp\Y_{ab}-2\nablao_{(a}  \omega_{b)}-2\omega_a\omega_b \\
\label{defabsRicci} & + 3\nablao_{(a} \sone_{b)} +\sone_a\sone_b-(\trP\bY)\U_{ab}+ 2P^{cd}\U_{d(a}\lp 2\Y_{b)c}+\F_{b)c}\rp.
\end{align} 
When the data is $(\phi,\rig)$-embedded in $(\M,g)$, the Ricci tensor $\textbf{\textup{Ric}}_g$ of $g$ and $\bs{\Rtensor}$ satisfy $\phi^{\star}(\textbf{\textup{Ric}}_g)=\bs{\Rtensor}$ \cite{manzano2023constraint}. Thus, $\bs{\Rtensor}$ encodes \textit{abstractly} the pull-back to $\N$ of the ambient Ricci tensor. This makes the constraint tensor a core object within the formalism. The following expressions for the contractions $\bs{\Rtensor}(n,\cdot)$ and $\bs{\Rtensor}(n,n)$ will be needed later \cite{manzano2023constraint}:
\begin{align}
\label{ConstTensror(n,-)} \Rtensor_{ab}n^b &=  -\nablao_{a}\Q  + \pounds_n\omega_a+ (\trP\bU)\omega_a - \nablao_a(\textup{tr}_P\bU) + P^{cd} \left ( \nablao_c\U_{ad} - 2 \U_{ac}\sone_{d} \right ),\\
\label{ConstTensror(n,n)} \Rtensor_{ab}n^an^b &= -n(\textup{tr}_P\bU)+(\trP\bU) \Q -P^{ad}P^{bc}\U_{ac}\U_{bd}.
\end{align}
Observe that the combination of \eqref{Riemosym(n,-)} and \eqref{ConstTensror(n,-)} yields
\begin{align}
\label{(R-Rico)(n,.)}\lp \Rtensor_{ab} -\Riemo_{(ab)}\rp n^b&=  \pounds_n\lp  \omega_a-\dfrac{1}{2}\sone_a\rp-\nablao_{a}\Q  + \lp\omega_a - \sone_a\rp\trP\bU.
\end{align}
Any null (metric) hypersurface data is subject to a built-in gauge freedom that accounts for the fact that, at the embedded level, the choice of rigging is non-unique \cite{mars2013constraint, mars2020hypersurface}. The precise definition works as follows.\  
Let $\hypdata$ be null hypersurface data, $z\in\mathcal{F}^{\star}(\mathcal{N})$ a no-where zero function and $V$ a vector field on $\N$. We call $z,V$ \textbf{gauge parameters}, and define the \textbf{gauge-transformed data} by
\begin{align}
	\label{gaugegamma&ell2} \hspace{-0.6cm}\mathcal{G}_{\lp z,V\rp}(\gamma)& \defi  \gamma ,\hspace{0.42cm}\mathcal{G}_{\lp z,V\rp}\lp\ellc\rp \defi z\lp\ellc+ \gamma \lp V,\cdot\rp\rp,\hspace{0.42cm}\mathcal{G}_{\lp z,V\rp}\big(\ell^{(2)}\big) \defi z^2\big( \ell^{(2)}+2\ellc\lp V\rp+ \gamma \lp V,V\rp\big),\\
	\label{gaugeY}\hspace{-0.6cm}\mathcal{G}_{\lp z,V\rp}(\bY)  & \defi z\bY+ \big( \ellc+\gamma(V,\cdot)\big)\otimes_s dz+\frac{z}{2}\lieo_{V} \gamma. 
\end{align}
The set $\{\N,\G_{(z,V)}(\gamma),\G_{(z,V)}(\ellc),\G_{(z,V)}(\elltwo),\G_{(z,V)}(\bY)\}$ 
is also null hypersurface data. Moreover, if $\hypdata$ is $(\phi,\rig)$-embedded in $(\M,g)$, then the gauge-transformed data is also embedded in $(\M,g)$ with same embedding but rigging $\G_{(z,V)}(\rig)\defi z\lp \rig+\phi_{\star}V\rp$.  
The transformations \eqref{gaugegamma&ell2}-\eqref{gaugeY} induce the following behaviour of $n$, $\bU$, $\bsone$, $\bomega$, $\nablao$ and 
$\ovnabla$ \cite{mars2020hypersurface,mars2023zcovariant,manzano2023constraint}:

\vspace{-0.6cm}

\noindent
\begin{minipage}[t]{0.5\textwidth}
\begin{align}
\label{gaugen}\G_{(z,V)}(n)=z^{-1}{n}{},
\end{align}
\end{minipage}
\hfill
\hfill
\begin{minipage}[t]{0.5\textwidth}
\begin{align}
\label{Uprime} \G_{(z,V)}(\bU)=z^{-1}{\bU}{},
\end{align}
\end{minipage}

\begin{minipage}[t]{0.62\textwidth}
\begin{align}
\label{gauge:sone}\hspace{-0.15cm}\G_{(z,V)}\big(\bsone\big)= \bsone+\frac{1}{2}\pounds_n\lp \gamma(V,\cdot)\rp+\frac{n(z)}{2z}\lp\ellc+\gamma(V,\cdot)\rp-\frac{dz}{2z},
\end{align}
\end{minipage}
\hfill
\begin{minipage}[t]{0.38\textwidth}
\begin{align}
\label{gauge:omega} \G_{(z,V)}(\bomega)=\bomega-\frac{dz}{z}+\bU(V,\cdot),
\end{align}
\end{minipage}

\noindent
\begin{minipage}[t]{0.66\textwidth}
\begin{align}
\label{gauge:nablao}
\G_{(z,V)}\big(\nablao\big)-\nablao=\frac{1}{2z}\lp  V\otimes \pounds_{zn}\gamma + n\otimes \lp \pounds_{zV}\gamma+2\ellc\otimes_s dz\rp \rp,
\end{align}
\end{minipage}
\hfill
\hfill
\begin{minipage}[t]{0.34\textwidth}
\begin{align}
\label{gauge:ovnabla}\G_{(z,V)}\big(\ovnabla\big)-\ovnabla= V\otimes \bU.
\end{align}
\end{minipage}

An interesting case occurs when a null hypersurface data 
is equipped with a \textbf{gauge-invariant vector field} $\ovkil$ along the degenerate direction of $\gamma$ (i.e.\ $\ovkil\in\Rad\gamma$). The gauge-invariance ensures that $\ovkil$ is a property of the hypersurface itself, and enables to describe mathematically relevant situations---e.g.\ horizons \cite{manzano2023field,mars2024transverseI,mars2024transverseII,mars2025KID}---abstractly. In particular, letting $\alpha\in\Fcal(\N)$ be the proportionality function between $\ovkil$ and $n$ (i.e.\ $\ovkil= \alpha n$), one can define the \textbf{surface gravity} $\kappa$ of $\ovkil$ as  \cite{manzano2023field}
\begin{equation}
\label{defkappaonN}\kappa\defi n(\alpha)+\alpha\Q.
\end{equation}
This quantity  is gauge-invariant. It is also well-defined and smooth everywhere on $\H$ (even at the zeroes of $\ovkil$), and in the embedded case it coincides with the standard notion of surface gravity, namely  
$\nabla_{\phi_{\star}\ovkil}\, \phi_{\star}\ovkil=\kappa \, \phi_{\star}\ovkil$ on $\{p\in\phi(\N)\spc\vert\spc\phi_{\star}\ovkil\vert_p\neq0\}$. 

\section{The Lie derivative of a connection}\label{sec:tensors:sigma}

Given a manifold endowed with a vector field $Z$ and a connection $D$, one can define the tensor ``\textbf{Lie derivative of $D$ along $Z$}", which we denote by $\SigmaZ\defi \pounds_{Z}D$. Tensor fields $\SigmaZ$ have been widely studied (e.g.\ in \cite{yano1957lie}, to which we refer for further details), albeit they are not so commonly used in the literature on hypersurface geometry. They, however, play a central role when the vector field $Z$ is privileged in some sense (e.g.\ when it defines a horizon), as they encode geometric properties of $Z$ and, at the same time, properties of the curvature tensor of $D$ (see e.g.\ equation \eqref{Sigma:Curvature} below). 

In the works  \cite{manzano2023field,mars2024transverseI,mars2024transverseII,mars2025KID}, various kinds of tensors ``Lie derivative of a connection" 
have been defined and analyzed, mostly for connections defined on the ambient space. In this paper, however, we take a slightly different approach and focus on tensors $\SigmaZ$ defined at the abstract level, i.e.\ for non-embedded hypersurface data. More specifically, 
we will introduce the tensor fields $\sigo\defi \pounds_n\nablao$ and $\ovsigmakil\defi \pounds_{\ovkil}\ovnabla$, where as before $\ovkil$ is a gauge-invariant vector field proportional to $n$. 
We will obtain their explicit expressions and analyze some of their properties. As we shall see, $\ovsigmakil$ 
plays an essential role in the study of abstract horizons. For instance, it is key for the derivation of the so-called generalized near-horizon equation (see Section \ref{sec:Cov:ME:General:Hyp}), an identity that generalizes the well-known near-horizon equation of isolated horizons \cite{ashtekar2000generic,  ashtekar2000isolated,ashtekar2002geometry, krishnan2002isolated,gourgoulhon20063+, jaramillo2009isolated} as well as the so-called master equation of multiple Killing horizons \cite{mars2018multiple}, to generic null hypersurfaces with any topology and with a possibly non-empty set of zeroes of $\ovkil$.

Consider a smooth manifold $\cv$ endowed with an affine connection $D$ and
a vector field $Z$. The $1$-contravariant, $2$-covariant tensor field $\SigmaZ \defi \pounds_{Z}D$ is defined by \cite{yano1957lie}
\begin{align}
\SigmaZ(X,W) \defi  \pounds_{Z}D_XW-D_X\pounds_{Z}W-D_{\pounds_ZX}W,\qquad \forall X,W\in\Gamma(T\cv).
\label{defderD}
\end{align}
This tensor is symmetric when $D$ is torsion-free, which we shall assume from now on. The tensor $\SigmaZ$ is related to the commutator of Lie and covariant derivatives by \cite{yano1957lie}
\begin{align}
\pounds_Z D_{\alpha} T_{\beta_1 \cdots \beta_\mathfrak{p}} = D_{\alpha} \pounds_Z T_{\beta_1 \cdots \beta_\mathfrak{p}} - \sum_{\mathfrak{i}=1}^\mathfrak{p} \SigmaZ^{\mu}{} _{\alpha \beta_\mathfrak{i}} T_{\beta_1 \cdots \beta_{\mathfrak{i}-1} \mu \beta_{\mathfrak{i}+1} \cdots \beta_\mathfrak{p}}. \label{commutation}
\end{align} 
When $T$ is a symmetric $(0,2)$-tensor field, performing a standard cyclic permutation of indices in equation \eqref{commutation} yields 
\begin{equation}
\label{Sigma:times:Sym:tensor}T_{\mu\sigma}\SigmaZ^{\sigma}{} _{\alpha \beta}  =\frac{1}{2}\lp \mathcal{S}_{\alpha\beta\mu}+\mathcal{S}_{\beta\mu\alpha} -\mathcal{S}_{\mu\alpha\beta}\rp, \quad\textup{where}\quad \mathcal{S}_{\alpha\beta\mu}\defi D_{\alpha} \pounds_Z T_{\beta\mu}-\pounds_Z D_{\alpha} T_{\beta\mu}. 
\end{equation}
The following two identities link 
the tensor $\SigmaZ$ and the curvature tensor ${R^{D}}\hspace{0.03cm}{}^{\mu}{}_{\beta\nu\alpha}$ of $D$  \cite{yano1957lie}:
\begin{align}
\label{Sigma:Curvature} \SigmaZ^{\mu}{}_{\alpha\beta}&=D_{\alpha} D_{\beta}Z^{\mu}+{R^{D}}\hspace{0.03cm}{}^{\mu}{}_{\beta\nu\alpha}Z^{\nu},\\
\label{Lie:Sigma:and:Curvature} \pounds_Z {R^D}{}^{\rho}{}_{\mu\alpha \beta}&=D_{\alpha}\SigmaZ^{\rho}{}_{\beta\mu}-D_{\beta}\SigmaZ^{\rho}{}_{\alpha\mu}.
\end{align}
For later use, we now provide the relation between the tensors $\SigmaZ$, $\mS[\alpha Z]$ for any scalar function $\alpha$.
\begin{lemma}\label{lem:sig:fZ:sigZ}
Let $\cv$ be a manifold, $Z$ a vector field, $D$ a torsion-free connection, and $\alpha$ a scalar function. For any $X,W\in\Gamma(T\cv)$, the vector fields $\mS[\alpha Z](X,W)$ and $\SigmaZ(X,W)$ are related by
\begin{align}
\label{sig:fZ:sigZ}\mS[\alpha Z](X,W)= \alpha\SigmaZ(X,W)+ \textup{\textbf{Hess}}^D(\alpha)(X,W) Z
+X(\alpha)D_WZ
+W(\alpha)D_X Z.
\end{align}
\end{lemma}
\begin{proof} 
In the computation, we shall use the identity $\pounds_{\alpha Z}W=\alpha\pounds_{Z}W-W(\alpha)Z$ several times. From the definition \eqref{defderD}, it follows
\begin{align*}
\mS[\alpha Z](X,W)
=&\spc \alpha\pounds_{Z}D_XW-(D_XW)(\alpha)Z - D_X\big( \alpha\pounds_{Z}W-W(\alpha)Z\big) -D_{\alpha\pounds_{Z}X-X(\alpha)Z}W\\
=&\spc \alpha\pounds_{Z}D_XW
-(D_XW)(\alpha)Z
-X(\alpha)\pounds_{Z}W
-\alpha D_X\pounds_{Z}W\\
&+X\big(W(\alpha)\big)Z
+W(\alpha)D_X Z
-\alpha D_{\pounds_{Z}X}W
+X(\alpha)D_{Z}W,
\end{align*}
which upon using \eqref{defderD}, the fact that $D$ is torsion-free (so that 
$-\pounds_Z W + D_Z W = D_W Z$) and the definition of the hessian operator of $D$, yields \eqref{sig:fZ:sigZ}.
\end{proof}
Now, any null metric hypersurface data $\metdata$ gives rise to a vector field $n$ that is privileged in the sense that its direction (but not its scale) remains unchanged by gauge transformations (cf.\ \eqref{gaugen}). It therefore makes sense to introduce the tensor $\sigo\defi \pounds_n\nablao$, which will in fact be useful later for the calculation of $\ovsigmakil$. Observe that we do not reflect the $n$-dependence in the notation for $\sigo$. The explicit expression of $\sigo$ in terms of metric data is obtained next. 
\begin{lemma}
\label{ExpreSigo}
Let $\metdata$ be null metric hypersurface data and define $\sigo\defi \pounds_{n}\nablao$. Then, 
\begin{align}
\label{sigo} 
{\sigo}{}^d{}_{ab} = &\spc n^d \left ( 2 \nablao_{(a} \sone_{b)} + n (\elltwo) \U_{ab}  \right ) + P^{dc} \left (\nablao_a \U_{bc} + \nablao_b \U_{ca} - \nablao_c \U_{ab} + 2 \sone_c \U_{ab}\right ). 
\end{align}
In particular, the contractions $\ell_f  {\sigo}{}^f{}_{ab}$ and $\gamma_{cf} {\sigo}{}^f{}_{ab}$ read
\begin{align}
\label{ell:sigo}\ell_f  {\sigo}{}^f{}_{ab} &=  2\nablao_{(a} \sone_{b)} + n (\elltwo) \U_{ab} + \elltwo \pounds_n \U_{ab}, \\
\gamma_{cf} {\sigo}{}^f{}_{ab} = &\spc 
\nablao_a \U_{bc} + \nablao_b \U_{ca} - \nablao_c \U_{ab} + \ell_c 
 \pounds_n \U_{ab}+ 2 \sone_c  \U_{ab} .
\label{gamsigo}
\end{align}
\end{lemma}
\begin{proof}
Particularizing \eqref{commutation} (resp.\ \eqref{Sigma:times:Sym:tensor}) for $Z=n$, $D=\nablao$ and  $T=\ellc$ (resp.\ $T=\gamma$)  yields
\begin{align}
\ell_f  {\sigo}{}^f{}_{ab} &= \nablao_a \pounds_n \ell_b - \pounds_n  \nablao_a \ell_b \defi  Q_{ab},\label{ellsigo}\\ 
\label{gammasigo}\gamma_{c f}{\sigo}{}^{f}{} _{ab}  &=\frac{1}{2}\lp \widetilde{\mathcal{Q}}_{abc}+\widetilde{\mathcal{Q}}_{bca} -\widetilde{\mathcal{Q}}_{cab}\rp,
\end{align}
where $\widetilde{\mathcal{Q}}_{abc}\defi \nablao_a \pounds_n \gamma_{bc} - \pounds_n \nablao_a \gamma_{bc}$. 
It only remains to determine $Q_{ab}$,  $\widetilde{\mathcal{Q}}_{abc}$ to obtain \eqref{sigo}-\eqref{gamsigo}.   
For $Q_{ab}$ we use \eqref{soneprop}, \eqref{nablaoll} and 
$(\pounds_n \bF)_{ab}  
= \frac{1}{2} (d \pounds_n \ellc)_{ab} 
= (d \bm{\sone})_{ab}= \nablao_a \sone_b - \nablao_b \sone_a$: 
\begin{align}
Q_{ab} &= \nablao_a \pounds_n \ell_b - \pounds_n \nablao_a \ell_b
= 2 \nablao_a \sone_b  - \pounds_n \left ( \F_{ab} - \elltwo \U_{ab} \right ) = 2\nablao_{(a} \sone_{b)} + n (\elltwo) \U_{ab} + \elltwo \pounds_n \U_{ab}.\label{expreQ}
\end{align}
This proves \eqref{ell:sigo}. For the second we combine $\pounds_n\gamma=2\bU$ (cf.\ \eqref{threetensors}) with  \eqref{soneprop} to get
\begin{align}
\nn \widetilde{\mathcal{Q}}_{abc} =&\spc 2\nablao_a\U_{bc}  
  + \pounds_n \left ( \ell_b \U_{ac} + \ell_c \U_{ab} \right )  = 2 \nablao_a \U_{bc}+ 2 \ell_{(b} \pounds_n \U_{c)a} + 2 (\pounds_n \ell_{(b}) \U_{c)a} \\
\stackbin{\eqref{nablaogamma}}=&\spc 2 \nablao_a \U_{bc} + 2 \ell_{(b}  \pounds_n \U_{c)a} + 4 \sone_{(b}\U_{c)a}. 
\label{expret}
\end{align}
Now, any  $(0,3)$-tensor of the form $t_{abc} =  2u_{(b} S_{c)a}$ with $S_{ca}$ symmetric  satisfies $t_{abc} + t_{bca} - t_{cab} = 2 u_c S_{ab}$. 
Inserting \eqref{expret} into 
\eqref{gammasigo} and using this property gives \eqref{gamsigo}.
To conclude the proof we use 
${\sigo}{}^{d}{}_{ab} = \delta^d_f {\sigo}{}^{f}{}_{ab} 
\stackbin{\eqref{prod4}}=\left ( P^{dc} \gamma_{cf} + n^d \ell_f \right ) {\sigo}{}^{f}{}_{ab} = P^{dc}\gamma_{cf} {\sigo}{}^{f}{}_{ab} + n^d Q_{ab}$. Replacing here \eqref{expreQ} and \eqref{gamsigo} yields \eqref{sigo} after using $P^{dc}\ell_c \stackbin{\eqref{prod3}}=  -\elltwo n^d$.
\end{proof}
As mentioned before, the case when a null hypersurface data $\hypdata$ 
admits a gauge-invariant vector field $\ovkil=\alpha n$ with $\alpha\in\Fcal(\N)$ is particularly relevant, as it allows one to analyze the geometry of horizons. In these circumstances, the vector field $\ovkil$ is privileged in the sense that it is a property of the abstract null hypersurface $\N$. Therefore, it is to be expected that a number of fundamental geometric properties of $\N$ are encoded in tensor fields of the form $\pounds_{\ovkil}D$ where $D$ is an abstractly defined connection. The fact that most horizons are totally geodesic null hypersurfaces (hence characterized by null data with $\bU=0$), together with the relation \eqref{nablaXYnablao} between the hypersurface connection $\ovnabla$ and the Levi-Civita derivative of $g$ in the embedded case, strongly suggests that the appropriate object for addressing horizon geometry is the tensor  $\ovsigmakil\defi \pounds_{\ovkil}\ovnabla$. Accordingly, we devote the remainder of this section to deriving its explicit expression, and obtaining its gauge behaviour. The analysis of $\ovsigmakil$ will allow us to define a tensor field $\calP^{\ovkil}$ (see Section \ref{sec:hor:data}) whose properties are useful to characterize horizons at an abstract level. We start by deriving the explicit form of $\ovsigmakil$. 
\begin{lemma}
Let $\hypdata$ be null hypersurface data and $\ovkil=\alpha n$, $\alpha\in\Fcal(\N)$ a gauge-invariant vector field. Then, the tensor field $\ovsigmakil\defi \pounds_{\ovkil}\ovnabla$ takes the form
\begin{align}
\label{ovsigmakil}
\ovsigmakil{}^c{}_{ab}=\alpha{\sigo}{}^c{}_{ab}+2P^{cd}(\nablao_{(a}\alpha)\U_{b)d}+\Big(\nablao_a\nablao_b\alpha+2(\nablao_{(a}\alpha)\omega_{b)}+n(\alpha)\Y_{ab}-\alpha(\pounds_{n}\bY)_{ab}\Big)n^c.
\end{align}
In particular, the contractions $\ell_c\ovsigmakil{}^c_{ab}$ and $\gamma_{fc}\ovsigmakil{}^c_{ab}$ read
\begin{align}
\label{ellc:ovsigmakil}\ell_c\ovsigmakil{}^c{}_{ab}=&\spc\nablao_a\nablao_b\alpha+2(\nablao_{(a}\alpha)\omega_{b)}+n(\alpha)\Y_{ab}+\alpha\lp 2\nablao_{(a}\sone_{b)}-\pounds_{n}\Y_{ab}+\pounds_n(\elltwo\U_{ab})\rp,\\
\label{gamma:ovsigmakil} \gamma_{fc}\ovsigmakil{}^c{}_{ab}
=&\spc \alpha \Big(\nablao_{a}\U_{bf}+\nablao_b\U_{af}-\nablao_f\U_{ab}+\ell_f\pounds_n\U_{ab}+2\sone_f\U_{ab}\Big)+2 (\nablao_{(a}\alpha)\U_{b)f}.
\end{align}
\end{lemma}
\begin{proof} 
Definition \eqref{def:ovnabla} establishes the following relation between $\ovsigmakil$ and ${\sigo}[\ovkil]\defi\pounds_{\ovkil}\nablao$ (not to be confused with ${\sigo}\defi \pounds_n\nablao$):
\begin{align*}
\ovsigmakil&= {\sigo}[\ovkil]
+n \otimes \Big(n(\alpha)\bY-\pounds_{\ovkil}\bY
\Big)  
={\sigo}[\ovkil]+n \otimes \Big(n(\alpha)\bY-\alpha\pounds_{n}\bY-2d\alpha\otimes_s\bs{\Yn}
\Big) .
\end{align*}
Applying now Lemma \ref{lem:sig:fZ:sigZ} in the r.h.s.\ we conclude
\begin{align*} 
\ovsigmakil{}^c{}_{ab}=&\spc\alpha{\sigo}{}^c{}_{ab}
+2(\nablao_{(a}\alpha)\nablao_{b)}n^c
+\Big(\nablao_a\nablao_b\alpha+n(\alpha)\Y_{ab}-\alpha(\pounds_{n}\bY)_{ab}
-2(\nablao_{(a}\alpha)\Yn_{b)}
\Big) \ n^c,
\end{align*}
which becomes \eqref{ovsigmakil} upon using \eqref{nablao:n} and $\bomega\defi \bsone-\bs{\Yn}$. 
Equation \eqref{ellc:ovsigmakil} follows \eqref{ell:sigo}  
after using 
$\ell_cP^{cd}\U_{bd}\stackbin{\eqref{prod3}}=-\elltwo n^d\U_{bd}\stackbin{\eqref{Un}}=0$  and $\ell_cn^c\stackbin{\eqref{prod2}}=1$. Finally, \eqref{gamma:ovsigmakil} is obtained from \eqref{gamsigo} after noticing that $\gamma(n,\cdot)\stackbin{\eqref{prod1}}=0$ and $\gamma_{fc}P^{cd}\U_{bd}\stackbin{\eqref{prod4}}=(\delta^{d}_{f}-n^{d}\ell_{f})\U_{bd}\stackbin{\eqref{Un}}=\U_{bf}$.
\end{proof}
The gauge transformations of $\ovsigmakil$ and the contraction $\ellc(\ovsigmakil)$ are as follows. 
\begin{lemma}
Let $\hypdata$ be null hypersurface data and $\ovkil=\alpha n$, $\alpha\in\Fcal(\N)$ a gauge-invariant vector field.  
Then,
\begin{align}
\label{gauge:ovsigmakil}\G_{(z,V)}\lp \ovsigmakil\rp=&\spc \ovsigmakil+\Big( \alpha\pounds_{n}V-V(\alpha)n\Big)\otimes\bU+\alpha V\otimes\pounds_{n}\bU,\\
\nn \G_{(z,V)}\lp \ellc \lp\ovsigmakil\rp \rp
=&\spc z\Big( \ellc\big( \ovsigmakil\big)+ \gamma(V,\ovsigmakil)\Big)+z\Big(\alpha\ellc\lp \pounds_nV\rp-V(\alpha)+\alpha\gamma(V,\pounds_nV)\Big)\bU\\
\label{gauge:ellc:ovsigmakil} &+\alpha z\Big(\ellc(V)+\gamma(V,V)\Big)\pounds_{n}\bU.
\end{align}
\end{lemma}
\begin{proof}
Since $\ovkil$ is gauge-invariant, $\G_{(z,V)}\big(\ovsigmakil\big)=\G_{(z,V)}\big(\pounds_{\ovkil}\ovnabla\big)=\pounds_{\ovkil}\big(\G_{(z,V)}\ovnabla\big)\stackbin{\eqref{gauge:ovnabla}}=\pounds_{\ovkil}\big(\ovnabla+V\otimes \bU\big)=\ovsigmakil
+\pounds_{\ovkil}V\otimes \bU+V\otimes \pounds_{\ovkil}\bU$, from where \eqref{gauge:ovsigmakil} follows after using $\pounds_{\ovkil}\bU=\pounds_{\alpha n}\bU\stackbin{\eqref{Un}}=\alpha \pounds_n\bU$. 
Equation \eqref{gauge:ellc:ovsigmakil}is a consequence of the gauge transformation of $\ellc$ (cf.\ \eqref{gaugegamma&ell2}).  
\end{proof}
For totally geodesic null hypersurfaces (i.e.\ when $\bU=0$), the tensor $\ovsigmakil$ is gauge-invariant while its contraction with $\ellc$ exhibits the remarkably simple gauge-behaviour $\G_{(z,V)}\lp \ellc \lp\ovsigmakil\rp \rp
= z\ellc\big( \ovsigmakil\big)$ (because $\gamma(\ovsigmakil,\cdot)=0$, by \eqref{gamma:ovsigmakil}). The restriction $\bU=0$, however, is quite limiting and excludes physically and mathematically relevant classes of horizons, for instance homothetic or conformal Killing horizons. In that case, the gauge properties of $\ovsigmakil$ become considerably more involved, which makes this tensor less suitable as a characterizing object. This motivates the introduction of the so-called  \isotensor $\bcalP^{\ovkil}$, a tensor field closely related to $\ovsigmakil$ but with better gauge-behaviour. This is done in the next section. 

\section{Deformation tensor of an ambient vector: the notion of $\mathcal{K}$-tuple}\label{sec:hor:data}

Typically, horizons embedded in a spacetime $(\M,g)$ are characterized by a privileged vector field $\kil\in\Gamma(T\M)$ that becomes null and tangent to a null hypersurface $\phi(\N)\subset\M$. In this setting, the so-called deformation tensor $\Kkil\defi \pounds_{\kil}g$ plays a fundamental role; in particular   
$\Kkil\vert_{\phi(\N)}$ is proportional to $g$ for homothetic or conformal Killing horizons, and vanishes for Killing horizons. Given the importance of $\Kkil$, it is natural to seek a \textit{purely abstract} notion of horizon data that captures its value at $\phi(\N)$. This leads us to the concept of \textbothdata, 
which we introduce next. 
\begin{definition}\label{def:phdata}
A set $\pmdata$ is called \textbf{\textpmdata} if $(i)$ $\metdata$ is null metric hypersurface data, and $(ii)$ $\alpha,\p\in\Fcal(\N)$,  $\bqone\in\Gamma(T^{\star}\N)$ are two scalar functions and a covector field with  gauge transformations
\begin{align}
\label{gauge:bqone:alpha:p:w}
\G_{(z,V)}(\alpha)= z\alpha,\quad\spc\spc\G_{(z,V)}(\bqone)= z\big(\bqone+2\alpha\bU(V,\cdot)\big),\quad\spc\spc \G_{(z,V)}(\p)=  z^2\big( \p+2\bqone(V)+2\alpha\bU(V,V)\big).
\end{align}  
If a \textpmdata is endowed with an additional symmetric $(0,2)$-tensor $\bY$, we call \textbf{\textphdata} the set $\phdata$.
\end{definition}
Given a \textpmdata, one can define a vector field $\ovkil\in\Rad\gamma$ and a function $\w$ by 
\begin{align}
\label{def:ovkil:and:w}\ovkil\defi \alpha n,\qquad \w\defi \bqone(n).
\end{align}
It then follows from \eqref{Un}, \eqref{gaugen} and \eqref{gauge:bqone:alpha:p:w} that both are gauge-invariant, i.e.\ 
\begin{align}
\G_{(z,V)}(\ovkil)=\ovkil,\qquad \G_{(z,V)}(\w)\defi \w.
\end{align}
Recall that the surface gravity $\kappa$ of $\ovkil$ was defined in \eqref{defkappaonN}, and that it is also gauge-invariant. We now establish a connection between the abstract notion of \textphdata and the deformation tensor of an \textit{ambient} vector field. This is done through the concept of embeddedness of a \textnonedata, which endows the latter with a geometric interpretation.
\begin{definition}\label{def:embedded:pdata}
Let $\pmdata$ be a \textpmdata, and define the vector $\ovkil\defi\alpha n$. Then,  
$\pmdata$ is $\bs{(\phi,\rig)}$\textbf{-embedded} in a Lorentzian manifold $(\M,g)$ if $\metdata$ is $(\phi,\rig)$-embedded in $(\M,g)$ and, in addition, there exists an extension $\kil\in\Gamma(T\M)$ of $\phi_{\star}\ovkil$ such that 
\begin{align}
\label{ext:eta:and:def:tensor}\phi^{\star}\Kkil=2\alpha\bU,\qquad \phi^{\star}\lp\Kkil(\rig,\cdot)\rp=\bqone,\qquad  \phi^{\star}\lp\Kkil(\rig,\rig)\rp=\p,
\end{align}
where $\Kkil\defi \pounds_{\kil}g$ is the deformation tensor of $\kil$. For a \textphdata $\phdata$, embeddedness requires that $\bY=\frac{1}{2}\phi^{\star}\lp\pounds_{\rig}g\rp$ as well.
\end{definition}
\begin{remark}
In the embedded case the quantities $\alpha\bU$, $\bqone$ and $\p$ encode all the components of the deformation tensor $\Kkil$ at $\phi(\N)$. Note also that the function $\w$, defined in \eqref{def:ovkil:and:w}, satisfies $\phi^{\star}\big(\Kkil(\rig,\phi_{\star}n)\big) = \w$.  
We emphasize that  
$\alpha\bU$, $\bqone$ and $\p$ are fully arbitrary and need not behave as if $\phi(\N)$ was an actual horizon.
\end{remark}
Naturally, two questions arise at this point. The first one is whether the gauge-behaviours \eqref{gauge:bqone:alpha:p:w} are consistent with changes in the rigging in the embedded case. 
{The second is under which conditions an embedded null metric data can be completed with a tuple $\{\alpha,\p,\bqone\}$ satisfying \eqref{gauge:bqone:alpha:p:w} to obtain an embedded \textpmdata. 
These issues are addressed in the following remark and proposition.}
\begin{remark}
Let $\pmdata$ be a \textpmdata, $(\phi,\rig)$-embedded in $(\M,g)$. Let $\kil$ be the extension of $\phi_{\star}\ovkil$ off $\phi(\N)$ according to Definition \ref{def:embedded:pdata}. 
The tensor $\Kkil$ is independent of the rigging, hence for \eqref{gauge:bqone:alpha:p:w} to be consistent with a change of rigging $\rig\rightarrow\rig'\defi z(\rig+\phi_{\star}V)$, it must hold
\begin{align}
\phi^{\star}\Kkil=\G_{(z,V)}\big(2\alpha\bU\big),\qquad \phi^{\star}\lp\Kkil(\rig',\cdot)\rp=\G_{(z,V)}\big(\bqone\big),\qquad  \phi^{\star}\lp\Kkil(\rig',\rig')\rp=\G_{(z,V)}\big(\p\big).
\end{align}
The first equation follows from $\G_{(z,V)}(\alpha)=z\alpha$ and \eqref{Uprime}, while for the second and third one finds
\begin{align*}
\phi^{\star}\lp\Kkil(\rig',\cdot)\rp=&\spc z\phi^{\star}\big(\Kkil(\rig+\phi_{\star}V,\cdot)\big)=z\big(\bqone+2\alpha\bU(V,\cdot)\big)\stackbin{\eqref{gauge:bqone:alpha:p:w}}=\G_{(z,V)}(\bqone),\\
\phi^{\star}\lp\Kkil(\rig',\rig')\rp=&\spc z^2\phi^{\star}\lp\Kkil(\rig+\phi_{\star}V,\rig+\phi_{\star}V)\rp= z^2\big( \p+2\bqone(V)+2\alpha\bU(V,V)\big)\stackbin{\eqref{gauge:bqone:alpha:p:w}}=\G_{(z,V)}(\p).
\end{align*}
We conclude that \eqref{gauge:bqone:alpha:p:w} is indeed compatible with transformations of rigging at the embedded level.
\end{remark}
\begin{proposition}\label{prop:extension:eta}
Let $\metdata$ be null metric hypersurface data $(\phi,\rig)$-embedded in a Lorentzian manifold $(\M,g)$. Consider any two functions $\alpha,\p\in\Fcal(\N)$ and any covector $\bqone\in\Gamma(T^{\star}\N)$  
satisfying \eqref{gauge:bqone:alpha:p:w}, and define $\ovkil\defi \alpha n$. Then  the \textpmdata $\pmdata$ can always be $(\phi,\rig)$-embedded in $(\M,g)$, and 
embeddedness is equivalent to 
the  
extension $\kil$ of $\phi_{\star}\ovkil$ off $\phi(\H)$ satisfying  
\begin{align}
\label{lie:rig:eta}\pounds_{\rig}\kil\stackbin{\phi(\N)}=\big(\bqone-d\alpha\big)(n)\rig
+\phi_{\star}\lp 
\frac{1}{2}\big(\p-\alpha n(\elltwo)\big) \, 
n
+P\big(\bqone-2\alpha\bsone-d\alpha, \, \cdot \, \big)\rp .
\end{align}
\end{proposition}
\begin{proof}
Equation \eqref{lie:rig:eta} always admits solutions because it only restricts the first order transverse derivative of $\kil$ at $\phi(\N)$. Hence, we only need to prove that \eqref{lie:rig:eta} holds if and only if $\pmdata$ is $(\phi,\rig)$-embedded in $(\M,g)$. Suppose first that the extension $\eta$ verifies \eqref{lie:rig:eta}, and define $\Kkil\defi \pounds_{\kil}g$. Then (recall $\bsone(n)=0$, \eqref{prod2}-\eqref{prod4}, and that $\phi_{\star}n$ is normal to $\phi(\N)$),
\begin{align*}
\phi^{\star}\Kkil&=\phi^{\star}\big(\pounds_{\kil}g\big)=\pounds_{\ovkil}\gamma=\alpha\pounds_n\gamma=2\alpha\bU,\\
\phi^{\star}\big(\Kkil(\rig,\cdot))(X)&=\phi^{\star}\big((\pounds_{\kil}g)(\rig,\cdot)\big)(X)=\phi^{\star}\Big(\pounds_{\kil}\big(g(\rig,\cdot)\big)+g(\pounds_{\rig}\kil,\cdot)\Big)(X)\\
&\stackbin{\eqref{emhd}}=\big(\pounds_{\ovkil}\ellc\big)(X)+\big(\bqone-d\alpha\big)(n)\ellc(X)+P^{ab}\big(\qone_b-2\alpha\sone_b-\nablao_b\alpha\big)\gamma_{ac}X^c\\
&\stackbin{\eqref{prod4}}=\Big(\alpha\pounds_{n}\ellc+d\alpha+\bqone(n)\ellc-d\alpha(n)\ellc\Big)(X)+(\delta^b_c-n^b\ell_c)\big(\qone_b-2\alpha\sone_b-\nablao_b\alpha\big)X^c\\
&\stackbin{\eqref{soneprop}}=\Big(2\alpha\bsone+d\alpha+\bqone(n)\ellc-d\alpha(n)\ellc+\bqone-2\alpha\bsone-d\alpha-\bqone(n)\ellc+d\alpha(n)\ellc\big)(X)\\
&=\bqone(X),\qquad \forall X \in\Gamma(T\mathcal{H}),\\
\phi^{\star}\big(\Kkil(\rig,\rig))&=\phi^{\star}\big((\pounds_{\kil}g)(\rig,\rig)\big)=\phi^{\star}\Big(\pounds_{\kil}\big(g(\rig,\rig)\big)+2g(\pounds_{\rig}\kil,\rig)\Big)\\
&\stackbin{\eqref{emhd}}=\alpha n(\elltwo)+2\elltwo\big(\bqone-d\alpha\big)(n)+\p-\alpha n(\elltwo)+2P^{ab}\ell_a \big(\qone_b-2\alpha\sone_b-\nablao_b\alpha\big)\\
&\stackbin{\eqref{prod3}}=\p+2\elltwo\big(\bqone-d\alpha\big)(n)-2\elltwo  \big(\bqone-d\alpha\big)(n)=\p.
\end{align*}
By construction $\kil$ fulfills \eqref{ext:eta:and:def:tensor}, which proves that the \textnonedata $\pmdata$ is $(\phi,\rig)$-embedded, as well as the sufficiency of \eqref{lie:rig:eta}. For the necessity, assume that $\pmdata$ is $(\phi,\rig)$-embedded in $(\M,g)$. Then, by Definition \ref{def:embedded:pdata} there exists an extension $\kil$ of $\phi_{\star}\ovkil$ so that \eqref{ext:eta:and:def:tensor} holds. Such extension necessarily verifies  \eqref{lie:rig:eta} by \cite[Thm.\ 4.1]{manzano2023field}.
\end{proof}
As mentioned before, the tensor $\ovsigmakil$ is not ideal for characterizing the horizon geometry due to its gauge transformation properties. 
To describe horizon geometry in more favorable terms, we introduce a new tensor field $\bcalP^{\ovkil}$, called the \isotensor. We first present its abstract definition together with an alternative expression for it. Then we obtain its gauge transformation, and explore its geometric meaning by considering the embedded case.
\begin{definition} \label{Def:calP:sigma:new}
Consider a \textphdata $\phdata$ with the  corresponding gauge-invariant vector field $\ovkil\defi \alpha n$. 
The \textbf{\isotensor} $\bcalP^{\ovkil}$ is the symmetric $(0,2)$-tensor field defined by
\begin{align}
\label{def:calP:ovsigmakil} \bcalP^{\ovkil}\defi \ellc\lp\ovsigmakil\rp-\frac{1}{2}\lp \alpha n\big(\elltwo\big)-\p\rp\bU-\alpha\elltwo\pounds_n\bU.
\end{align}
\end{definition}
\begin{remark}
By inserting \eqref{ellc:ovsigmakil} into \eqref{def:calP:ovsigmakil}, one obtains the following alternative expression for $\bcalP^{\ovkil}$: 
\begin{align}
\label{calP:usual:expression}\calP^{\ovkil}_{ab}= \nablao_a\nablao_b\alpha+2(\nablao_{(a}\alpha)\omega_{b)}+2\alpha\nablao_{(a}\sone_{b)}+n(\alpha)\Y_{ab}-\alpha(\pounds_{n}\bY)_{ab}+\frac{1}{2}\lp \alpha n(\elltwo)+\p\rp\U_{ab}.
\end{align}
\end{remark}
\begin{remark}
The name \isotensor is justified by the fact that, when $\bU=0$ 
(i.e.\ when the abstract hypersurface has vanishing second fundamental form), it holds
\[
{\ovsigmakil}{}^c{}_{ab}\stackbin[\eqref{ellc:ovsigmakil}]{\eqref{ovsigmakil}}= \alpha \lp {\sigo}{}^c{}_{ab}-2n^c\nablao_{(a}\sone_{b)}\rp +n^c\ell_d{\ovsigmakil}{}^d{}_{ab}
\stackbin{\eqref{sigo}}= n^c\ell_d{\ovsigmakil}{}^d{}_{ab}
\stackbin{\eqref{def:calP:ovsigmakil}}=n^c\calP^{\ovkil}_{ab}. 
\]
In this case, $\bcalP^{\ovkil}=0$ is therefore equivalent to $\ovsigmakil\defi \pounds_{\ovkil}\ovnabla=0$, which is the defining condition of an isolated horizon \textup{\cite{ashtekar2002geometry}}.   
Thus, $\bcalP^{\ovkil}$ accounts for the ``level of non-isolation" of a \textnonedata.
\end{remark}
\begin{lemma} 
Consider a \textphdata $\phdata$ and let $\ovkil\defi \alpha n$ be its associated gauge-invariant vector field. The gauge transformation of the \isotensor $\bcalP^{\ovkil}$ is given by
\begin{align}
\label{gauge:calP}\G_{(z,V)}(\calP^{\ovkil}_{ab})=z\calP^{\ovkil}_{ab}
+zV^c\Big( \qone_c\U_{ab}
+ 2\nablao_{(a}(\alpha\U_{b)c})
-\nablao_c(\alpha\U_{ab})\Big).
\end{align}
\end{lemma}
\begin{proof}
We denote $\G_{(z,V)}$-transformed quantities with a prime. We first derive the gauge transformation of the last two terms in \eqref{def:calP:ovsigmakil}.\ Using the transformation laws  
\eqref{gaugegamma&ell2}, \eqref{gaugen}-\eqref{Uprime} and \eqref{gauge:bqone:alpha:p:w}, one gets
\begin{align*}
\Big( \alpha' n'\big(\elltwo{}'\big)-\p'\Big)\bU'=&\spc \frac{1}{z} \lp \alpha n\Big(z^2\big( \ell^{(2)}+2\ellc\lp V\rp+ \gamma \lp V,V\rp\big)\Big)-z^2\big( \p+2\bqone(V)+2\alpha\bU(V,V)\big)\rp \bU\\
=&\spc  
z\lp \alpha n(\elltwo)-\p\rp \bU+2\bigg( \alpha n(z)\big( \ell^{(2)}+2\ellc\lp V\rp+ \gamma \lp V,V\rp\big)\\
&+z\alpha n\Big( \ellc\lp V\rp+ \frac{1}{2}\gamma \lp V,V\rp\Big)-z\bqone(V)-z\alpha\bU(V,V)\bigg) \bU,\\
=&\spc  
z\lp \alpha n(\elltwo)-\p\rp \bU+2\bigg( \alpha n(z)\big( \ell^{(2)}+2\ellc\lp V\rp+ \gamma \lp V,V\rp\big)\\
&+z\alpha \Big( 2\bsone\lp V\rp+\ellc\lp \pounds_nV\rp+\gamma(\pounds_nV,V)\Big)-z\bqone(V)\bigg) \bU,\\
\alpha'\elltwo{}'\pounds_{n'}\bU'=&\spc z^3\alpha \big( \ell^{(2)}+2\ellc\lp V\rp+ \gamma \lp V,V\rp\big)\pounds_{z^{-1}n}(z^{-1}\bU)\\
= &\spc \alpha \big( \ell^{(2)}+2\ellc\lp V\rp+ \gamma \lp V,V\rp\big)\big( z\pounds_{n}\bU-n(z)\bU\big).
\end{align*}
It only remains to gauge-transform the right-hand side of \eqref{def:calP:ovsigmakil} by using these two expressions as well as the transformation law for $\ellc \lp\ovsigmakil\rp$ in  \eqref{gauge:ellc:ovsigmakil} to obtain 
\begin{align*}
\bcalP^{\ovkil}{}'=&\spc z\bcalP^{\ovkil}+z\gamma(V,\ovsigmakil)+z\lp \bqone-d\alpha-2\alpha \bsone\rp(V)\hspace{0.05cm} \bU-z\alpha \ellc\lp V\rp\pounds_{n}\bU,
\end{align*}
which upon substituting $\gamma(\cdot,\ovsigmakil)$ according to \eqref{gamma:ovsigmakil} gives \eqref{gauge:calP}.
\end{proof}
A notion of ``\isotensor{}" was presented in \cite[Def.\ 1]{mars2012stability} 
for an embedded  
null hypersurface $\Hemb$ with vanishing second fundamental form.  
While in \cite{mars2012stability} this object was not directly related to the Lie derivative of the spacetime Levi-Civita connection $\nabla$, by comparison with \eqref{Sigma:Curvature} it follows that it actually coincides with the tensor $\pounds_{\nu}\nabla$  
restricted to directions tangent to $\Hemb$. 
The notion of \isotensor in \cite{mars2012stability}, however, is not abstract, applies only in the case when $\bU=0$, and involves a null vector tangent to $\Hemb$ which is not allowed to vanish therein. In the more recent work \cite[Eq.\ (5.6)]{manzano2023field}, the authors provide a generalization of the \isotensor 
to null hypersurfaces with arbitrary tensor $\bU$ embedded in a spacetime 
$(\M,g)$ admitting a privileged vector field $\kil$ that is null and tangent along the hypersurface, and which is allowed to vanish. Nevertheless, the definition of \cite{manzano2023field} was still non-abstract, as it arose from computing certain components of the tensor $\pounds_{\kil}\nabla$,  and it involved the pull-back $\phi^{\star}(\pounds_{\rig}\Kkil)$.  
In the present paper we have developed a framework, based on the notion of \textnonedata, that enables us to define and study $\bcalP^{\ovkil}$ \textit{at a purely abstract level}, and endow it with a clear geometric meaning 
that relates it to the tensor $\pounds_{\kil}\nabla$ in the embedded case, as we shall see next. 

Consider a \textphdata $\phdata$ that is $(\phi,\rig)$-embedded in a Lorentzian manifold $(\M,g)$. Let $\kil$ be an extension 
of $\phi_{\star}\ovkil$ off $\phi(\N)$ such that \eqref{ext:eta:and:def:tensor} holds. 
We want to relate the \isotensor $\bcalP^{\ovkil}$ with $\Kkil$ via the tensor 
$\Sigmakil\defi \pounds_{\kil}\nabla$. Particularizing \eqref{Sigma:times:Sym:tensor} for $Z=\kil$, $D=\nabla$, $T=g$ and multiplying by $g^{\mu\lambda}$, one  
finds \cite{manzano2023field}
\begin{equation}
\label{Sigma:eta:emb}\Sigmakil^{\lambda}{} _{\alpha \beta}  =\frac{1}{2}g^{\mu\lambda}\big( \nabla_{\alpha} \Kkil_{\beta\mu}+\nabla_{\beta} \Kkil_{\alpha\mu}-\nabla_{\mu} \Kkil_{\alpha\beta}\big).
\end{equation} 
The tensor $\Sigmakil$ is independent of the choice of rigging, hence its value at $\phi(\N)$ must be gauge-invariant. The expression of $\Sigmakil$ at $\phi(\N)$ was obtained in \cite[Lem.\ 5.2]{manzano2023field}, and reads
\begin{align}
\nn \Sigmakil&(e_a,e_b)\stackbin{\phi(\N)}= \Big( \big( \w-n(\alpha) \big)\U_{ab}- \alpha  (\pounds_n \bU)_{ab} \Big)\rig\\
\nn &+ \Big( \nablao_{a}\nablao_{b}\alpha +2(\nablao_{(a}\alpha)\omega_{b)}+ 2\alpha\nablao_{(a}\sone_{b)}+n(\alpha) \Y_{ab}-\alpha(\pounds_{n}\bY)_{ab}+\frac{1}{2}\big( \alpha n(\elltwo)+ \p\big)\U_{ab}\Big)\phi_{\star}n\\
\label{Sigma:general:xprssn} &+P^{cd}\lp 2(\nablao_{(a}\alpha) \U_{b)c}	+ \U_{ab}( \qone_c-\nablao_c\alpha )+\alpha( \nablao_a\U_{bc}+ \nablao_b\U_{ca}- \nablao_c\U_{ab}) \rp e_d ,
\end{align}
where $\{e_a\}$ is a local basis of $\Gamma(T\phi(\N))$, $\alpha,\p,\bqone$ are defined by \eqref{ext:eta:and:def:tensor}, 
and $\w\defi \bqone(n)$. Observe that the factor in parentheses in the second line of \eqref{Sigma:general:xprssn} is precisely $\bcalP^{\ovkil}$ (cf.\ \eqref{calP:usual:expression}). In fact, combining \eqref{prod3}, \eqref{Un} and the identity $n^c\big( \nablao_a\U_{bc}+\nablao_b\U_{ca}-\nablao_c\U_{ab}\big)=-(\pounds_n\bU)_{ab}$, it is straightforward to prove \cite{manzano2023field}
\begin{align} 
\label{pullback:sigmakil}\phi^{\star}\lp g(\rig,\Sigmakil)\rp=\bcalP^{\ovkil}.
\end{align}
This general identity, along with the gauge-invariance of $\Sigmakil$ when 
$\Kkil=2 c g$, $c\in\mathbb{R}$,  
suggests that $\bcalP^{\ovkil}$ is a suitable candidate for studying horizon geometry. We elaborate on this in the next remark.  
\begin{remark}\label{rem:horizons} 
Consider a \textpmdata $\pmdata$ with $\alpha\bU=c\gamma$, $\bqone=2c\ellc$ and $\p=2c\elltwo$, where $c$  
is a gauge-invariant constant. 
Then, a direct calculation 
based on \eqref{nablaogamma} and \eqref{gauge:calP}
gives 
\begin{align*}
\alpha\lp \G_{(z,V)}(\calP^{\ovkil}_{ab})
-z\calP^{\ovkil}_{ab}\rp&=zV^c\Big( \alpha\qone_c\U_{ab}
+ 2\alpha\nablao_{(a}(\alpha\U_{b)c})
-\alpha\nablao_c(\alpha\U_{ab})\Big)\\
&=czV^c\Big( 2c\ell_c \gamma_{ab}
+ \alpha\big(2\nablao_{(a}\gamma_{b)c}
-\nablao_c\gamma_{ab}\big)\Big)
\stackbin{\eqref{nablaogamma}}= 
2cz\ellc(V)\Big( c \gamma_{ab}
- \alpha\U_{ab}\Big)=0.
\end{align*}
When, in addition, the set of zeroes of $\alpha$ has empty interior\footnote{This restriction on the zero-set of $\alpha$ holds for smooth homothetic, conformal Killing and Killing horizons
.}\  in $\N$, it follows that $\G_{(z,V)}(\bcalP^{\ovkil})=z\bcalP^{\ovkil}$. 

Now, the data $\metdata$ is always $(\phi,\rig)$-embeddable in a Lorentzian manifold $(\M,g)$ \textup{\cite[Thm.\ 4.2]{mars2024transverseII}}, and 
by Proposition \ref{prop:extension:eta} we know that $\pmdata$ can also be embedded in $(\M,g)$. This means that there exists an extension $\kil\in\Gamma(T\M)$ of $\phi_{\star}\ovkil$ so that $\Kkil=2c g$ on $\phi(\N)$ (cf.\ \eqref{emhd}, Definition \ref{def:embedded:pdata}). Thus, $\phi(\N)$ is either a zeroth-order homothetic  
horizon or a zeroth-order Killing horizon, depending on whether 
$c\neq0$ or $c=0$. By comparison with \eqref{gauge:ovsigmakil}-\eqref{gauge:ellc:ovsigmakil}, it is immediate to see that the gauge-behaviours of $\ovsigmakil$ and $\ellc\big(\ovsigmakil\big)$ are significantly more involved than that of $\bcalP^{\ovkil}$ 
in the homothetic 
case where $\bU\neq0$. In particular, when 
the extension $\kil$ is an actual homothetic/Killing vector,    
$\Sigmakil$ vanishes (by \eqref{Sigma:eta:emb})  
and consequently so does the \isotensor $\bcalP^{\ovkil}$ (by \eqref{pullback:sigmakil}).  
We therefore conclude that 
abstract homothetic/Killing horizons are characterized by \textnonedata{}s with $\alpha\bU=c\gamma$ and $\bcalP^{\ovkil}=0$, conditions that are both gauge-independent, as $\alpha\bU$ is gauge-invariant (cf.\ \eqref{Uprime}, \eqref{gauge:bqone:alpha:p:w}) and $\bcalP^{\ovkil}$ simply rescales under gauge transformations. 
\end{remark}

We conclude the section by providing the relation between the \isotensor $\bcalP^{\ovkil}$ and the first transverse derivative of the deformation tensor for a given embedded \textnonedata. 
\begin{lemma}\label{lem:order_one}
Let $\phdata$ be a \textphdata $(\phi,\rig)$-embedded in a Lorentzian manifold $(\M,g)$, and   $\kil$ be the corresponding extension of $\phi_{\star}\ovkil\defi \phi_{\star}(\alpha n)$ off $\phi( \N)$. Then,
\begin{align}
\label{id:pi:and:lie_rig:def:tensor}\bcalP^{\ovkil}=-\frac{1}{2}\phi^{\star}\lp \pounds_{\rig}\Kkil-\pounds_{\mathcal{W}}g\rp,
\end{align}
where $\Kkil\defi \pounds_{\kil}g$ and 
$\mathcal{W}^{\beta}\defi g^{\beta\mu}\Kkil_{\mu\nu}\rig^{\nu}$.
\end{lemma}
\begin{proof}
Let $\Sigmakil\defi \pounds_{\kil}\nabla$. Then $\rig_{\lambda}\Sigmakil^{\lambda}{}_{\alpha\beta}\stackbin{\eqref{Sigma:eta:emb}}=\frac{1}{2}\rig^{\rho}\lp 2\nabla_{(\alpha}\Kkil_{\beta)\rho}-\nabla_{\rho}\Kkil_{\alpha\beta}\rp$, therefore 
\begin{align*}
\rig_{\lambda}\Sigmakil^{\lambda}{}_{\alpha\beta}&
=\rig^{\mu}\nabla_{(\alpha}\Kkil_{\beta)\mu}-\frac{1}{2}\rig^{\mu}\nabla_{\mu}\Kkil_{\alpha\beta}
=\rig^{\mu}\nabla_{(\alpha}\Kkil_{\beta)\mu}-\frac{1}{2}\pounds_{\rig}\Kkil_{\alpha\beta}+(\nabla_{(\alpha}\rig^{\mu})\Kkil_{\beta)\mu}\\
&=-\frac{1}{2}\pounds_{\rig}\Kkil_{\alpha\beta}+\nabla_{(\alpha}(\Kkil_{\beta)\mu}\rig^{\mu}).
\end{align*}
This expression can be rewritten as 
\begin{align}
\label{id:pi:and:lie_rig:def:tensor:2}g(\rig,\Sigmakil)=-\frac{1}{2}\lp \pounds_{\rig}\Kkil-\pounds_{\mathcal{W}}g\rp
\end{align}
after noticing that $2\nabla_{(\alpha}\theta_{\beta)}=\pounds_{g^{\sharp}(\bs{\theta},\cdot)}g_{\alpha\beta}$ for any $\bs{\theta}\in\Gamma(T^{\star}\M)$. Taking the pull-back of \eqref{id:pi:and:lie_rig:def:tensor:2} to $\N$
gives \eqref{id:pi:and:lie_rig:def:tensor} because of \eqref{pullback:sigmakil}.
\end{proof}

\section{A near-horizon equation on general null hypersurfaces}\label{sec:Cov:ME:General:Hyp}

Given a Lorentzian manifold $(\M,g)$, an embedded null hypersurface $\Hemb$ 
is said to be \textbf{totally geodesic} if its second fundamental form $\btsff^{\nu}(X,W) \defi  g(\nabla_X \nu, W), \, X,W \in \Gamma(T\Hemb)$, vanishes on $\Hemb$. Common examples of totally geodesic null hypersurfaces are non-expanding, weakly isolated and isolated horizons \cite{ashtekar2000generic,ashtekar2000isolated, ashtekar2002geometry,krishnan2002isolated, gourgoulhon20063+, jaramillo2009isolated}, or Killing horizons \cite{wald1984general,frolov2012black}. When $\btsff^{\nu}=0$, any null (not necessarily everywhere non-zero) vector field $\kil$ tangent to $\Hemb$ is proportional to $\nu$, and therefore also verifies
$g(\nabla_X\kil,W)=0$.  
Thus,  $\nabla_{X} \kil\vert_{\nullhyp}$ is null and tangent to $\nullhyp$. Denoting by $\mathcal{S}$ the set of points of $\nullhyp$ where $\kil$ vanishes, we can define a one-form $\bs{w}[\kil]$ on $\nullhyp\setminus\mathcal{S}$ by 
\begin{equation}
\label{def:bkilone}\nabla_X\kil\spc\stackbin{\nullhyp\setminus\mathcal{S}}=\spc\bs{w}[\kil](X)\kil\qquad \forall X\in\Gamma(T\nullhyp\setminus\mathcal{S}).
\end{equation}
Observe that the surface gravity  
of $\kil$ is given by  
$\bs{w}[\kil](\kil)$  on $\wt\N\setminus\mathcal{S}$. Note also that, in general, $\bs{w}[\kil]$ cannot be extended to $\mathcal{S}$ is because the right-hand side of \eqref{def:bkilone} vanishes, while 
$\nabla_X\kil$ need not be zero. If the null hypersurface $\Hemb$ comes from $(\phi,\rig)$-embedded null hypersurface data $\hypdata$ with $\bU=0$, the combination of \eqref{nablaXYnablao} and \eqref{def:bkilone} yields
\begin{align}
\label{relation:omega:and:bkilone}\bs{w}[\nu](X)\nu=\nabla_X\nu=\nablao_{X}n-\bY(X,n)\nu 
\stackbin{\eqref{nablao:n}}=\big( \sone -\bs{r}\big)\big(X\big)\nu\defi \bomega(X)\nu,
\end{align}
so $\bomega$ is the abstract counterpart of $\bs{w}[\nu]$. When $\Hemb$ admits a cross-section $S$, $\bs{w}[\nu]\vert_S$ is related to the torsion one-form $\bs{\mathfrak{s}}$ of $S$ (viewed as a codimension-two spacelike submanifold), as follows:  
\begin{align*}
\forall N\in\Gamma(TS)^{\perp}  \textup{ such that } 
g(N,\nu)\vert_S\neq0, 
\textup{ and }  
\forall X\in\Gamma(TS), \quad 
\bs{\mathfrak{s}}(X)\defi -\frac{g(\nabla_X \nu,N)}{g(\nu,N)}=-\bs{w}[\nu](X).
\end{align*}
For the purposes of this section, two results from the literature are particularly relevant: the \textbf{near-horizon equation} \cite[Eq.\ (5.3)]{ashtekar2002geometry} of  isolated horizons, and the \textbf{master equation} \cite[Eq.\ (60)]{mars2018multiple} of multiple Killing horizons.  
An isolated horizon \cite{ashtekar2002geometry} is a totally geodesic embedded null hypersurface $\Hemb_{\textup{iso}}$ with product topology $S\times\mathbb{R}$, and which  
admits a null, no-where zero tangent vector field $\kil\in\Gamma(T\Hemb_{\textup{iso}})$ such that $[\pounds_{\kil},\mathcal{D}]=0$ on $\Hemb_{\textup{iso}}$, where $\mathcal{D}$ is the connection induced from the Levi-Civita derivative of the ambient spacetime\footnote{In the language of the hypersurface data formalism $\mathcal{D}\equiv\ovnabla$, see Section \ref{sec:prelim}. Moreover, When $\bU=0$ the connection $\mathcal{D}$ is independent of the rigging \cite{ashtekar2002geometry}. This translates into $\ovnabla$ being gauge invariant (cf.\ \eqref{gauge:ovnabla}).}.   
Under the additional assumption of $\kil$ having vanishing surface gravity  
on $\Hemb_{\textup{iso}}$, the near-horizon equation is an identity that holds on a 
cross-section $S$ of $\Hemb_{\textup{iso}}$, which relates the one-form $\bs{w}[\kil]$, the Ricci tensor of the induced metric $h$ on $S$, and the spacetime Ricci tensor along tangential directions. On the other hand, a multiple Killing horizon $\Hemb_{\textup{m}}$ is a null hypersurface which is a Killing horizon with respect to at least two (linearly independent) Killing vector fields $\kil_1,\kil_2$ \cite{mars2018multiple}. Again, when a multiple Killing horizon admits a cross-section $S$, the master equation is an identity on $S$ which relates second derivatives of the proportionality function $\alpha$ between $\kil_1$ and $\kil_2$ on $\Hemb_{\textup{m}}$, the one-form $\bs{w}[\kil_1]$, the Ricci tensor of $h$, and the pull-back to $S$ of the spacetime Ricci tensor. 

While these equations have proven useful in a variety of situations (see e.g.\ \cite{ashtekar2002geometry,Kunduri_2009,kunduri2013classification,dunajski2017einstein,dobkowski2018near,alaee2019existence,moore2023bakry} and references therein), they hold under very specific conditions. For isolated horizons, the hypersurface is totally geodesic, has product topology, typically satisfies energy conditions and/or the Einstein equations, and the near-horizon equation applies only on cross-sections where $\kil$ does not vanish. For multiple Killing horizons, on the other hand, one needs two Killing vectors $\kil_1,\kil_2$ sharing a common horizon which in addition must admit a cross-section, and even in this setting the master equation is only valid at points when both $\kil_1$ and $\kil_2$ are non-zero. 

At this stage, a natural question is whether one can obtain analogous equations for more general null hypersurfaces, e.g.\ for horizons with arbitrary topology, for null hypersurfaces containing points where $\kil$ vanishes, or for other types of horizons (e.g.\ homothetic or conformal Killing horizons). This is the purpose of this section. More specifically, we shall obtain a generalized form of the near-horizon equation that holds for any null hypersurface $\N$ equipped with a privileged tangent null vector field $\ovkil$. This equation is fully covariant and holds on the whole $\N$ (not only on a cross-section, which in fact need not exist). It also provides useful information about the constancy of the surface gravity of $\ovkil$. 
\begin{theorem}\label{thm:master:equation}
Consider a \textphdata $\phdata$ with gauge-invariant vector field $\hspace{0.05cm}\ovkil\defi \alpha n$ and surface gravity $\kappa\defi n(\alpha)+\alpha \Q$.   
Let $\Riemo_{ab}$ be the Ricci tensor of $\nablao$ and define 
$\bs{\Rtensor}$, $\bcalP^{\ovkil}$ by \eqref{defabsRicci}, \eqref{def:calP:ovsigmakil} respectively. Then, the following identities hold:
\begin{align}
\nonumber 0 =&\spc \nablao_{a}\nablao_{b}\alpha +2\omega_{(a}\nablao_{b)}\alpha +\frac{\alpha}{2} \lp 2\nablao_{(a}\omega_{b)}+2\omega_{a}\omega_{b}+\mathcal{R}_{ab}-\Riemo_{(ab)}\rp \\
\nn &+ \lp \kappa +\frac{\alpha}{2} \trP\bU \rp\Y_{ab}+\frac{\alpha}{2}(\trP\bY)\U_{ab}- 2\alpha P^{cd}\U_{c(a} \Y_{b)d} -\calP^{\ovkil}_{ab}\\
\label{ME:calP} &  +\frac{\alpha}{2} \lp \nablao_{(a}\sone_{b)}-\sone_a\sone_b\rp
 +\frac{1}{2}\lp \alpha n(\elltwo)+\p\rp\U_{ab}- \alpha P^{cd}\U_{c(a} \F_{b)d},\\
0=&\spc\nablao_a\kappa+\alpha\lp\Rtensor_{ab}n^b-(\trP\bU)\omega_a+\nablao_a(\trP\bU)\rp -P^{bc}\nablao_{b}\lp\alpha\U_{ac}\rp-\calP^{\ovkil}_{ab}n^b,\label{ME(n,.):calP}\\
\label{ME(n,n):calP} 0 =&\spc n(\kappa)-\calP^{\ovkil}_{ab}n^an^b.
\end{align}
\end{theorem}
\begin{proof}
We shall repeatedly use the equalities
$\bsone(n)=0$, $\bU(n,\cdot)=0$, $\bs{\omega}(n)\stackbin{\eqref{defY(n,.)andQ}}=\Q$ and $n(\alpha)+\alpha\Q\stackbin{\eqref{defkappaonN}}=\kappa$. 
Equation \eqref{ME:calP} follows immediately from multiplying \eqref{defabsRicci} by $\frac{\alpha}{2}$ and using \eqref{calP:usual:expression} to replace $\alpha\pounds_n\bY$ in terms of $\bcalP^{\ovkil}$. 
Contracting \eqref{ME:calP} with $n^b$ and using \eqref{(R-Rico)(n,.)}  
along with  
\begin{align*}
\nn n^b\nablao_{a}\nablao_{b}\alpha=&\spc  \nablao_{a}(n^b\nablao_{b}\alpha )-(\nablao_{a}n^b)\nablao_{b}\alpha\stackbin{\eqref{nablao:n}}= \nablao_{a}(n(\alpha ))-n(\alpha ) \sone_a-P^{bf}\U_{af}\nablao_{b}\alpha \\
=&\spc\nablao_{a}\kappa-\alpha \nablao_{a}\Q-\Q\nablao_{a}\alpha -(\kappa-\alpha \Q)\sone_a-P^{bc}\U_{ac}\nablao_{b}\alpha ,\\
n^b\nablao_{(a}\sone_{b)}\stackbin{\eqref{n:nablao:theta:sym}}=&\spc 
\dfrac{1}{2}\pounds_{n}\sone_a-P^{bc}\U_{ac}\sone_{b},\quad \quad n^b\nablao_{(a}\omega_{b)}\stackbin{\eqref{n:nablao:theta:sym}}=\frac{1}{2}\pounds_{n}\omega_a  +  \frac{1}{2} \nablao_{a}\Q  -\Q \sone_a    -   P^{bc}\omega_{b}\U_{ac}, 
\end{align*}
one obtains 
\begin{align}
\label{mid:eq:2109}0=\nablao_a\kappa-\alpha\nablao_{a}\Q+\alpha\pounds_n\omega_a-P^{bc}\U_{ac}\lp\nablao_b\alpha+2\alpha\sone_b\rp-\calP^{\ovkil}_{ab}n^b 
\end{align}
which, upon using \eqref{ConstTensror(n,-)} to rewrite $\pounds_n\bomega$ in terms of $\bs{\Rtensor}(n,\cdot)$, becomes \eqref{ME(n,.):calP}. Multiplying 
\eqref{mid:eq:2109}
by $n^a$ gives \eqref{ME(n,n):calP} since $\bomega(n)=\Q$.
\end{proof}
Equation \eqref{ME:calP} is an abstract, covariant identity involving hypersurface data $\hypdata$, curvature terms ($\Rtensor_{ab}$ and $\Riemo_{ab}$), derivatives of $\alpha$, the surface gravity $\kappa$ of $\ovkil$, and the quantities $\p$ and $\bcalP^{\ovkil}$. Its derivation requires prescribing a \textphdata, i.e.\ not only hypersurface data but also $\alpha$, $\p$, and $\bqone$. However, $\p$ only appears multiplied by $\bU$, and $\bqone$ only plays a role when one considers gauge transformations of the \isotensor $\bcalP^{\ovkil}$ (recall \eqref{def:calP:ovsigmakil}-\eqref{gauge:calP}). Furthermore, the quantities $\bU$ and $\bcalP^{\ovkil}$ vanish for abstract Killing horizons---see the corresponding discussion in Section \ref{sec:hor:data} as well as in \cite[Sect.\ 6]{manzano2023field}---, in which case \eqref{ME:calP} only depends on hypersurface data plus the function $\alpha$.

As mentioned before, \eqref{ME:calP} generalizes in several directions the already known forms of near-horizon \cite[Eq.\ (5.3)]{ashtekar2002geometry} and master equations \cite[Eq.\ (60)]{mars2018multiple}.
First, \eqref{ME:calP} holds \textit{everywhere} on the hypersurface and not only on a specific cross-section (in fact, we are not assuming that such a section exists). Secondly, \eqref{ME:calP} does not require any topological assumption on the hypersurface apart from the 
requirement that it is two sided (or equivalently that the radical of $\gamma$ admits a global nowhere zero section $n$ \cite{mars2024abstract}).  
This is also a significant difference with respect to the works \cite{ashtekar2000generic, ashtekar2000isolated,ashtekar2002geometry,krishnan2002isolated,  gourgoulhon20063+, jaramillo2009isolated,mars2018multiple} where the hypersurface is assumed to have a product topology $S\times\mathbb{R}$.  Regarding the vector field $\ovkil$ we have also kept full generality by allowing $\ovkil$ to vanish on $\N$. Observe also that, by not restricting $\bU$, $\p$ and $\bqone$ in any way, we have enforced neither any specific extension of $\ovkil$ off $\phi(\N)$ (cf.\ Proposition \ref{prop:extension:eta}), nor any particular form of the deformation tensor of $\K^{\eta}$ in case the data happens to be embedded. 
Finally, note that \eqref{ME:calP} has been obtained for a fully arbitrary constraint tensor $\bs{\Rtensor}$. We have imposed neither energy conditions nor field equations of any kind. In view of the considerations above, we call \eqref{ME:calP}  
\textbf{generalized near-horizon equation}. 
The name is justified because \eqref{ME:calP} holds \textit{everywhere on $\N$}, and whenever a cross-section $S$ exists it also yields an equation therein, i.e.\ the pullback of \eqref{ME:calP} to $S$.

Let us conclude this section with an immediate corollary of equations \eqref{ME(n,.):calP}-\eqref{ME(n,n):calP} regarding the constancy of the surface gravity $\kappa$ of $\ovkil$. The function $\kappa$ is constant 
along the null generators 
if and only if $\bcalP^{\ovkil}(n,n)=0$, and it is constant everywhere on $\N$ if and only if $$\calP^{\ovkil}_{ab}n^b=\alpha\lp\Rtensor_{ab}n^b-(\trP\bU)\omega_a+\nablao_a(\trP\bU)\rp -P^{bc}\nablao_{b}\lp\alpha\U_{ac}\rp.$$ 
Thus, for any \textnonedata satisfying $\bU=0$ and $\bcalP^{\ovkil}=0$ the surface gravity $\kappa$ is automatically constant along the generators, whereas it is constant everywhere on $\N$ provided the zero set of $\alpha$ in $\N$ has empty interior and $\bs{\Rtensor}(n,\cdot)=0$.  

\section{Horizons in Lorentzian manifolds: existence results}\label{sec:existence:results}

Given a \textnonedata with $\bU=0$, the generalized near-horizon equation \eqref{ME:calP} takes the simplified form
\begin{align}
\nn 0 =&\spc \kappa \Y_{ab} -\calP^{\ovkil}_{ab}+\nablao_{a}\nablao_{b}\alpha +2\omega_{(a}\nablao_{b)}\alpha \\
\label{ME:calP:U=0}&+\alpha \lp \nablao_{(a}\omega_{b)}+\omega_{a}\omega_{b}+\frac{1}{2}\mathcal{R}_{ab}-\frac{1}{2}\Riemo_{(ab)}+\frac{1}{2} \nablao_{(a}\sone_{b)}-\frac{1}{2}\sone_a\sone_b\rp.
\end{align}
This expression allows one to 
determine the tensor $\bY$ algebraically from metric data, the one-form $\bomega$, the function $\alpha$ and the tensors $\bs{\Rtensor}$ and $\bcalP^{\ovkil}$, provided $\kappa\neq0$ everywhere on $\N$. The purpose of this section is to identify the conditions under which a data set of the form $\{ \N,\gamma,\ellc,\elltwo,\alpha,\bomega,\bcalP^{\ovkil},\bs{\Rtensor}\}$ with $\kappa\neq0$ can be $(\phi,\rig)$-embedded in a Lorentzian manifold $(\M,g)$ in such a way that $\frac{1}{2}\phi^{\star}(\pounds_{\rig}g)$ coincides with the tensor $\bY$ prescribed by the generalized near-horizon equation \eqref{ME:calP:U=0}. One of the motivations for studying this problem comes from the fact that the full transverse expansion of the spacetime metric at Killing horizons with $\kappa\neq0$ is uniquely determined by the order zero of the metric (i.e.\ $\metdataa$, cf.\ \eqref{emhd}), a one-form that satisfies certain evolution equation along the generators, and all transverse derivatives of the ambient Ricci tensor at the horizon \cite{mars2024transverseI,mars2024transverseII}. In this context, it is not possible to provide the entire tensor $\bY$ from the outset, but only a posteriori can one reconstruct it from  first-order transverse derivatives of the spacetime metric (cf.\ \eqref{YtensorEmbDef}). It is therefore reasonable to follow a similar procedure starting with data $\{\H,\gamma,\ellc,\elltwo,\bomega,\alpha,\bcalP^{\ovkil},\bs{\Rtensor}\}$, thereby extending the previous result to totally geodesic null hypersurfaces without any restrictions on the \isotensor. In particular, we need to understand which properties need to be imposed on this data to guarantee that the tensor $\bY$ defined by \eqref{ME:calP:U=0} agrees, after all, with the tensor $\frac{1}{2}\phi^{\star}(\pounds_{\rig}g)$ of the constructed spacetime.

The section is organized as follows. We begin by proving a number of identities for null data with $\bU=0$, which will be used later (Lemmas \ref{lem:sect:6:1}, \ref{lem:B:times:n} and \ref{lem:identity:lie:constraint}). Then, we obtain (Theorem \ref{thm:exist:pi:R}) the necessary and sufficient conditions for the existence of a Lorentzian manifold $(\M,g)$ constructed from $\{ \N,\gamma,\ellc,\elltwo,\alpha,\bomega,\bcalP^{\ovkil},\bs{\Rtensor}\}$, with no topological restriction on $\N$. 
We continue by addressing the case when $\N$ admits a cross-section, prescribing suitable data therein and again proving existence of $(\M,g)$ (Theorem \ref{thm:existence:sections}). The section ends with four remarks on these results (Remarks \ref{rem:1} to 
\ref{rem:4}). 

\begin{lemma}\label{lem:sect:6:1}
Let $\metdata$ be null metric hypersurface data with $\bU=0$, $\alpha\in\Fcal(\N)$ a scalar function and $\bs{\theta}$ a covector field.\ Define the vector field $\ovkil\defi \alpha n$. Then,
\begin{align}
\label{Lie_n:nablao:cov:sym}\pounds_n(\nablao_{(a}\theta_{b)})=&\spc\nablao_{(a}(\pounds_n\theta_{b)})-2\bs{\theta}(n)\nablao_{(a}\sone_{b)},\\
\nonumber \pounds_{\alpha n}\big(\alpha\nablao_{(a}\theta_{b)} \big)=&\spc \alpha\bigg(  n(\alpha)\nablao_{(a}\theta_{b)}+\nablao_{(a}\big(\alpha\pounds_n\theta_{b)}\big)\\
\label{Lie:alpha:n:nablao:theta}&+ (\nablao_{(a}\alpha)\nablao_{b)}\big(\bs{\theta}(n)\big)-2\bs{\theta}(n)\lp \alpha\nablao_{(a}\sone_{b)} +(\nablao_{(a}\alpha)\sone_{b)}\rp \bigg),\\ 
\label{lie_n:Riemo}\pounds_{n}\Riemo_{(ab)}=&\spc\pounds_n(\nablao_{(a}\sone_{b)})-2\sone_{(a}\pounds_n\sone_{b)},\\
\label{Lie:eta:nablao:nablao:alpha}\pounds_{\ovkil}\big(\nablao_a\nablao_b\alpha\big)=&\spc \alpha\nablao_a\nablao_b\big(n(\alpha)\big)+2(\nablao_{(a}\alpha)\nablao_{b)}\big(n(\alpha)\big)-2n(\alpha)\lp \alpha\nablao_{(a}\sone_{b)}+(\nablao_{(a}\alpha)\sone_{b)}\rp.
\end{align}
\end{lemma}
\begin{proof}
Observe that in the case $\bU=0$ the tensor $\sigo$ simplifies to (cf.\ \eqref{sigo})
\begin{align}
\label{Sigmacero:U=0}{\sigo}{}^c{}_{ab}=&\spc 2n^c\nablao_{(a}\sone_{b)}.
\end{align}
Equation \eqref{Lie_n:nablao:cov:sym} is the symmetrization of \eqref{commutation} in the case $Z=n$, $D=\nablao$ and $T=\bs{\theta}$. For \eqref{Lie:alpha:n:nablao:theta} we make use of the following identity for any function $\alpha\in\Fcal(\N)$ and any symmetric $(0,2)$-tensor $T_{ab}$:\ 
\begin{align}
\label{lie:eta:lie:n}\pounds_{\alpha n}T_{ab}=\alpha\pounds_n T_{ab}+2(\nablao_{(a}\alpha)T_{b)c}n^c.
\end{align}
Therefore, 
\begin{align*}
\pounds_{\alpha n}\big(\alpha\nablao_{(a}\theta_{b)} \big)
=&\spc \alpha n(\alpha)\nablao_{(a}\theta_{b)}+\alpha\lp \alpha\pounds_n \nablao_{(a}\theta_{b)}+(\nablao_{a}\alpha)n^c\nablao_{(b}\theta_{c)}+(\nablao_{b}\alpha)n^c\nablao_{(a}\theta_{c)}\rp \\
\stackbin{\eqref{n:nablao:theta:sym}}=&\spc \alpha\lp  n(\alpha)\nablao_{(a}\theta_{b)}+ \alpha\pounds_n \nablao_{(a}\theta_{b)}\rp +\alpha (\nablao_{(a}\alpha)\lp\pounds_n\theta_{b)}+\nablao_{b)}\big(\bs{\theta}(n)\big)-2\bs{\theta}(n)\sone_{b)}\rp \\
=&\spc \alpha\bigg(  n(\alpha)\nablao_{(a}\theta_{b)}+\nablao_{(a}\big(\alpha\pounds_n\theta_{b)}\big)-2\alpha\bs{\theta}(n)\nablao_{(a}\sone_{b)} + (\nablao_{(a}\alpha)\lp\nablao_{b)}\big(\bs{\theta}(n)\big)-2\bs{\theta}(n)\sone_{b)}\rp \bigg),
\end{align*}    
from where \eqref{Lie:alpha:n:nablao:theta} follows at once.  
To prove \eqref{lie_n:Riemo} we first particularize \eqref{Lie:Sigma:and:Curvature} for $Z=n$, $D=\nablao$,  
and find $\pounds_n\Riemo_{(ab)}=\nablao_c{\sigo}{}^c{}_{ab}-\nablao_{(a}{\sigo}{}^c{}_{b)c}$.
Now, using \eqref{Sigmacero:U=0}, $\bsone(n)=0$ and the fact that $n^c\nablao_b\sone_c=-\sone_c\nablao_b n^c\stackbin{\eqref{nablao:n}}=-\sone_b\bsone(n)=0$ and $n^c \nablao_c s_b = \pounds_n \sone_b -(\nablao_b n^c)\sone_{c}\stackbin{\eqref{nablao:n}}= \pounds_n \sone_b$, one obtains
\begin{align*}
\nablao_a {\sigo}{}^c{}_{bc}
=&\spc \nablao_a\lp n^c\nablao_{c}\sone_{b}\rp 
=\nablao_a\lp \pounds_n\sone_{b}\rp=\pounds_n(\nablao_a\sone_b),
\end{align*}
where in the last step we have used \eqref{commutation} for $Z=n$, $D=\nablao$ and $T=\bsone$. On the other hand, since $2\big(\nablao_bn^c\big)\big(\nablao_{(a}\sone_{c)}\big)=2\sone_bn^c\big(\nablao_{(a}\sone_{c)}\big)=\sone_bn^c\nablao_{c}\sone_{a}=\sone_b\pounds_n\sone_a$ and $\nablao_cn^c=0$, 
it follows
\begin{align*}
\nablao_c{\sigo}{}^c{}_{ab}=&\spc 
2\Big( \pounds_n\big(\nablao_{(a}\sone_{b)}\big) -\big(\nablao_an^c\big)\big(\nablao_{(b}\sone_{c)}\big)-\big(\nablao_bn^c\big)\big(\nablao_{(a}\sone_{c)}\big)\Big)
=2\Big( \pounds_n\big(\nablao_{(a}\sone_{b)}\big)-\sone_{(a}\pounds_n\sone_{b)}\Big).
\end{align*}
From the last two identities, it is immediate to obtain \eqref{lie_n:Riemo}. Finally, to demonstrate \eqref{Lie:eta:nablao:nablao:alpha} we first use \eqref{lie:eta:lie:n} and then apply \eqref{Lie_n:nablao:cov:sym} with $\theta_b=\nablao_b\alpha$ (note that $\nablao_a\nablao_b\alpha$ is symmetric in $ab$):
\begin{align*}
\pounds_{\ovkil}\big(\nablao_a\nablao_b\alpha\big)=&\spc \alpha\pounds_n (\nablao_a\nablao_b\alpha)+2n^c(\nablao_{(a}\alpha)(\nablao_{b)}\nablao_c\alpha)\\
=&\spc \alpha\pounds_n (\nablao_a\nablao_b\alpha)+2(\nablao_{(a}\alpha)\lp \nablao_{b)}\big(n(\alpha)\big)-(\nablao_c\alpha)(\nablao_{b)}n^c)\rp \\
=&\spc \alpha\lp \nablao_{(a}(\pounds_n\nablao_{b)}\alpha)-2n(\alpha)\nablao_{(a}\sone_{b)}\rp +2(\nablao_{(a}\alpha)\lp \nablao_{b)}\big(n(\alpha)\big)-n(\alpha)\sone_{b)}\rp.
\end{align*}
This yields \eqref{Lie:eta:nablao:nablao:alpha} after noticing that $\pounds_n\nablao_b\alpha=\nablao_b\big(n(\alpha) \big)$ and reorganizing terms.
\end{proof}

\begin{lemma}\label{lem:B:times:n}
Let $\metdata$ be null metric hypersurface data with $\bU=0$. Consider a function $\alpha\in\Fcal(\N)$ and a covector field $\bs{\rho}\in\Gamma(T^{\star}\N)$. Define
\begin{align*}
\mathcal{B}_{ab}\defi \nablao_{a}\nablao_{b}\alpha
+2\rho_{(a}\nablao_{b)}\alpha
+\alpha\lp  \nablao_{(a}\rho_{b)}+\rho_{a}\rho_{b}-\frac{1}{2}\Riemo_{(ab)}+\frac{1}{2}\nablao_{(a}\sone_{b)}-\frac{1}{2}\sone_a\sone_b\rp. 
\end{align*}
Then,
\begin{align}
\label{Tensor:B:times:n}
\mathcal{B}_{ab}n^b=&\spc 
\nablao_a(n(\alpha))
-\lp n(\alpha)
+\alpha\bs{\rho}(n)\rp \lp  \sone_a-\rho_{a}\rp
+\bs{\rho}(n)\nablao_{a}\alpha
+\frac{\alpha}{2}\pounds_{n}\rho_a  
+\frac{\alpha}{2} \nablao_{a}\big(\bs{\rho}({n})\big).  
\end{align}
\end{lemma}
\begin{proof}
Equation \eqref{Tensor:B:times:n} follows from the combination of  

\vspace{-0.6cm}

\noindent
\begin{minipage}[t]{0.5\textwidth}
	\begin{align}
		\label{hess:alpha:(n,.)}
		n^b\nablao_a\nablao_b\alpha\stackbin{\eqref{nablao:n}}=&\spc\nablao_a(n(\alpha))-n(\alpha)\sone_a,
	\end{align}
\end{minipage}
\hfill
\hfill
\begin{minipage}[t]{0.5\textwidth}
	\begin{align}
		\label{sym:ricco:n:proof}\Riemo_{(ab)}n^b\stackbin{\eqref{Riemosym(n,-)}}=&\spc\frac{1}{2}\pounds_n\sone_a,
	\end{align}
\end{minipage}
the identity \eqref{n:nablao:theta:sym} for $\bs{\theta}=\bs{\rho},\bsone$, and the fact that 
$\bsone(n)=0$.
\end{proof}
\begin{lemma}\label{lem:identity:lie:constraint}
Let $\hypdata$ be null hypersurface data with $\bU=0$ that is $(\phi,\rig)$-embedded in a Lorentzian manifold $(\M,g)$. Consider a gauge-invariant vector $\ovkil\defi \alpha n\in\Gamma(T\N), \, \alpha\in\Fcal(\N)$, and define $\kappa\defi n(\alpha)+\alpha\bomega(n)$ where $\bomega$ is given by \eqref{defY(n,.)andQ}. 
Assume that $\wt\bcalP$ and $\wt{\bs{\Rtensor}}$ are two symmetric $(0,2)$-tensor fields  satisfying
\begin{align}
\label{Pi(n,dot)=Lie:omega:new}\wt\bcalP(n,\cdot)=&\spc\alpha\pounds_n\bomega+\bomega(n)d\alpha+d\big(n(\alpha)\big),\\
\nn 
\kappa\Y_{ab}=&\spc
\wt{\calP}_{ab}-\nablao_{a}\nablao_{b}\alpha -2\omega_{(a}\nablao_{b)}\alpha \\
\label{kappa:Y:for:omega} &-\alpha \lp \nablao_{(a}\omega_{b)}+\omega_{a}\omega_{b}+\frac{1}{2}\wt{\mathcal{R}}_{ab}-\frac{1}{2}\Riemo_{(ab)}+\frac{1}{2}\nablao_{(a}\sone_{b)}-\frac{1}{2}\sone_a\sone_b\rp .
\end{align}
Then, the following identity holds:
\begin{align}
\nn 0=&\spc\Big(\pounds_{\ovkil}+\alpha\bomega(n)\Big)\wt\calP_{ab}-\alpha\nablao_{(a}\big(\wt\calP_{b)c}n^c\big)-\alpha\wt\bcalP(n,n)\Y_{ab}\\
\label{Lie:Pi:general}&-2\big(\alpha\omega_{(a}+\nablao_{(a}\alpha\big)\wt\calP_{b)c}n^c-\frac{\alpha}{2}\pounds_{\ovkil}\wt\Rtensor_{ab}+\frac{\alpha\kappa}{2}\lp \phi^{\star}\big(\textup{\textbf{Ric}}_g\big){}_{ab}-\wt\Rtensor_{ab}\rp.
\end{align}
\end{lemma}
\begin{proof}
Throughout the proof we shall make use of the identities 
\begin{align}
\label{lie:kil:theta:theta:new}\pounds_{\alpha n}(\alpha \hspace{0.05cm} \theta_a\theta_b)=&\spc \alpha \Big(n(\alpha)\theta_a\theta_b+2\theta_{(a}\pounds_{\alpha n}\theta_{b)}\Big),\hspace{0.7cm}\quad \forall\bs{\theta}\in \Gamma(T^{\star}\N),\\
\label{lie:eta:omega:nablao:alpha:new}\pounds_{\ovkil}\big(2\omega_{(a}\nablao_{b)}\alpha\big)=&\spc 2\big(\pounds_{\ovkil}\omega_{(a}\big)\big(\nablao_{b)}\alpha\big)+2\omega_{(a}\lp \alpha\nablao_{b)}\big( n(\alpha)\big)+n(\alpha)\nablao_{b)}\alpha \rp ,\\ 
\label{lie:eta:tilde:R:new}\pounds_{\ovkil}\lp \frac{\alpha}{2}\wt{\Rtensor}_{ab}\rp=&\spc \frac{\alpha}{2}\Big(n(\alpha)\wt{\Rtensor}_{ab}+\pounds_{\ovkil}\wt{\Rtensor}_{ab}\Big),\\
\label{lie:eta:Riemo:new} \pounds_{\ovkil}\lp \frac{\alpha}{2}\Riemo_{(ab)}\rp\stackbin[\eqref{lie:eta:lie:n}]{\eqref{lie_n:Riemo}}=&\spc \frac{\alpha}{2}\Big(n(\alpha)\Riemo_{(ab)}+\alpha\pounds_{n}\big(\nablao_{(a}\sone_{b)}\big)+\big(\nablao_{(a}\alpha-2\alpha\sone_{(a}\big)\pounds_n\sone_{b)}\Big),
\end{align}
First, notice that the contraction of \eqref{Pi(n,dot)=Lie:omega:new} with $n$ gives 
\begin{align}
\label{Pi(n,n):new}\wt\bcalP(n,n)&=\alpha n\big(\bomega(n)\big)+\bomega(n)n(\alpha)+n\big(n(\alpha)\big)=n\big(\alpha\bomega(n)+n(\alpha)\big)=n(\kappa).
\end{align}
Equation \eqref{Lie:Pi:general} is derived by computing the $\pounds_{\ovkil}$-derivative of \eqref{kappa:Y:for:omega}. For the l.h.s., one finds
\begin{align}
\nn\pounds_{\ovkil}\big(\kappa\Y_{ab}\big)=&\spc \alpha n(\kappa)\Y_{ab}+\kappa\pounds_{\ovkil}\Y_{ab}\stackbin[]{\eqref{Pi(n,n):new}}=\alpha \wt{\bcalP}(n,n)\Y_{ab}+2\kappa(\nablao_{(a}\alpha)(\sone_{b)}-\omega_{b)})+\alpha\kappa\pounds_n\Y_{ab}\\
\nn\stackbin{\eqref{defabsRicci}}=&\spc \alpha \wt{\bcalP}(n,n)\Y_{ab}+2\kappa(\nablao_{(a}\alpha)(\sone_{b)}-\omega_{b)})\\
\label{mid:eq:sdgyf}&+\alpha\kappa\lp \frac{1}{2}\Riemo_{(ab)}-\frac{1}{2}\phi^{\star}\big(\textup{\textbf{Ric}}_g\big){}_{ab}
- \bomega(n)\Y_{ab}-\nablao_{(a}\omega_{b)} -\omega_{a}\omega_{b}+  \frac{3}{2}\nablao_{(a}\sone_{b)}+\frac{1}{2}\sone_{a}\sone_{b}\rp.
\end{align}
For the r.h.s., we combine the definition of $\kappa$ with the general identities 
\eqref{Lie:eta:nablao:nablao:alpha}, 
\eqref{Lie:alpha:n:nablao:theta},   
\eqref{lie:kil:theta:theta:new} for $\bs{\theta}=\bomega,\bsone$, and 
\eqref{lie:eta:omega:nablao:alpha:new}-\eqref{lie:eta:Riemo:new}. Writing each derivative 
$\pounds_{n}\bomega$
in terms of $\wt\bcalP(n,\cdot)$ according to \eqref{Pi(n,dot)=Lie:omega:new}, and equating the left and right hand sides, \eqref{Lie:Pi:general} follows after several simplifications take place.
\end{proof}
Having the identities of Lemmas \ref{lem:sect:6:1}, \ref{lem:B:times:n} and \ref{lem:identity:lie:constraint} are at our disposal, we can now obtain the necessary and sufficient conditions for the existence of a Lorentzian manifold where a data set of the form $\{\N,\gamma,\ellc,\elltwo, \bomega,\alpha,{\bcalP},{\bs\Rtensor}\}$ is embedded. This is done in the next theorem. 
\begin{theorem}\label{thm:exist:pi:R}
Let $\{\N,\gamma,\ellc,\elltwo, \wt\bomega,\alpha,\wt {\bcalP},\wt{\bs\Rtensor}\}$ be a tuple consisting of 
\begin{itemize}
\item[$(i)$] null metric hypersurface data $\metdata$ with $\bU=0$,
\item[$(ii)$] a covector $\wt\bomega\in\Gamma(T^{\star}\H)$ with gauge behaviour $\G_{(z,V)}(\wt\bomega)=\wt\bomega-z^{-1}dz$, 
\item[$(iii)$] a function $\alpha\in\Fcal(\H)$ such that $\G_{(z,V)}(\alpha)=z\alpha$ and whose set of zeros has empty interior, and
\item[$(iv)$] two symmetric $(0,2)$-tensor fields $\wt {\bcalP},\wt{\bs\Rtensor}$ with gauge behaviours 
$$\G_{(z,V)}(\wt {\bcalP})=z\wt {\bcalP},\qquad
\G_{(z,V)}\big(\wt{\bs\Rtensor}\big)=\wt{\bs\Rtensor}.$$
\end{itemize}
Define 
$\ovkil\defi \alpha n$ and  
$\kappa\defi n(\alpha)+\alpha\wt\bomega(n)$, and suppose that $\kappa\neq0$ everywhere. 
Then,\ there exists a Lorentzian manifold $(\M,g)$,\ an embedding $\phi:\N\hookrightarrow\M$ and a rigging $\rig$ such that
\begin{itemize}
\item[\textit{1}.] $\metdata$ is $(\phi,\rig)$-embedded in $(\M,g)$;
\item[\textit{2}.] $\wt{\bomega}=\bsone-\frac{1}{2}\phi^{\star}\big(\pounds_{\rig}g\big)(n,\cdot)$;
\item[\textit{3}.] $\Sigmakil(\phi_{\star}X,\phi_{\star}W)=\wt\bcalP(X,W)\phi_{\star}n \, \, \forall X,W\in\Gamma(T\N)$, where $\kil$ is any extension of $\phi_{\star}\ovkil$ off $\phi(\N)$;
\item[\textit{4}.] $\phi^{\star}(\textup{\textbf{Ric}}_g)=\bs{\wt\Rtensor}$;
\end{itemize}
if and only if the following conditions hold:
\begin{align}
\label{Pi(n,dot)=Rtensor}\wt\bcalP(n,\cdot)=&\spc\alpha\wt{\bs\Rtensor}(n,\cdot)+d\kappa,\\
\label{Pi(n,dot)=Lie:omega}\wt\bcalP(n,\cdot)=&\spc\alpha\pounds_n\wt\bomega+\wt\bomega(n)d\alpha+d\big(n(\alpha)\big),\\
\nonumber 0=&\spc\Big(\pounds_{\ovkil}+\alpha\wt\bomega(n)\Big)\wt\calP_{ab}-\alpha\nablao_{(a}\big(\wt\calP_{b)c}n^c\big)-\alpha\wt\bcalP(n,n)\Y_{ab}\\
\label{Lie:Pi}&-2\big(\alpha\wt\omega_{(a}+\nablao_{(a}\alpha\big)\wt\calP_{b)c}n^c-\frac{\alpha}{2}\pounds_{\ovkil}\wt\Rtensor_{ab}.
\end{align}
where $\bY$ is the symmetric $(0,2)$-tensor defined by 
\begin{align}
\label{kappa:Y:prev} \hspace{-0.2cm}\kappa\Y_{ab}=
\wt{\calP}_{ab}-\nablao_{a}\nablao_{b}\alpha -2\wt{\omega}_{(a}\nablao_{b)}\alpha -\alpha \lp \nablao_{(a}\wt{\omega}_{b)}+\wt{\omega}_{a}\wt{\omega}_{b}+\frac{1}{2}\wt{\mathcal{R}}_{ab}-\frac{1}{2}\Riemo_{(ab)}+\frac{1}{2}\nablao_{(a}\sone_{b)}-\frac{1}{2}\sone_a\sone_b\rp.
\end{align}
Moreover, the pair $(\phi,\rig)$ satisfies $\frac{1}{2}\phi^{\star}\big(\pounds_{\rig}g\big)=\bY$, hence the null hypersurface data $\hypdata$ is also $(\phi,\rig)$-embedded in $(\M,g)$. 
\end{theorem}
\begin{remark}
As already mentioned, we are interested in allowing for the existence of zeroes of $\ovkil$ on $\N$. However, we do want to exclude the possibility that $\ovkil$ vanishes on open subsets of $\N$, since on such subsets the available information would be very limited (e.g., $\alpha$ and $\bcalP^{\ovkil}$ would be identically zero), and some of the previous results would no longer hold. 
Hence, we require that the zero set of $\alpha$ has empty interior on $\N$ (cf.\ $(iii)$). This assumption is consistent with the well-known fact that Killing vectors cannot vanish on any codimension-one subset of a spacetime, unless they are identically zero.
\end{remark}
\begin{remark}
For the quantities $\alpha$, $\wt\bomega$, $\wt\bcalP$ and $\wt{\bs{\Rtensor}}$ to satisfy \textit{1}.-\textit{4}.\ a posteriori, one needs to prescribe the gauge-behaviours of $(ii)$-$(iv)$ from the outset (recall \eqref{gauge:omega}, \eqref{gauge:bqone:alpha:p:w}, \eqref{gauge:calP} and the fact that, for embedded data, the constraint tensor is the pull-back of the ambient Ricci tensor, hence it is gauge-invariant). 
\end{remark}
\begin{remark}\label{rem:Raychaudhuri}
The contraction of \eqref{Pi(n,dot)=Rtensor} with $n$ yields $\wt\bcalP(n,n)=\alpha\wt{\bs{\Rtensor}}(n,n)+n(\kappa)$. On the other hand, the contraction of \eqref{Pi(n,dot)=Lie:omega} with $n$ is (cf.\ \eqref{Pi(n,n):new}) 
$\wt\bcalP(n,n)=n(\kappa)$. Since the set of zeroes of $\alpha$ has empty interior on $\N$, it follows that \eqref{Pi(n,dot)=Rtensor}-\eqref{Pi(n,dot)=Lie:omega} imply $\wt{\bs{\Rtensor}}(n,n)=0$. This is consistent with the Raychaudhuri equation for totally geodesic null hypersurfaces, namely \eqref{ConstTensror(n,n)} with $\bU=0$.
\end{remark}
\begin{proof}
Suppose that \textit{1}.--\textit{4}. hold and define 
\begin{align*}
    \mathbb{Y}\defi \frac{1}{2}\phi^{\star}\big(\pounds_{\rig}g\big),\qquad \mathbb{w}\defi \bsone-\mathbb{Y}(n,\cdot),\qquad \mathbb{k}_n\defi -\mathbb{Y}(n,n),\qquad \mathbb{k}\defi n(\alpha)+\alpha \mathbb{k}_n.
\end{align*}
Consider any extension $\kil$ of $\phi_{\star}\ovkil$ away from $\phi(\N)$, and define the deformation tensor $\Kkil\defi\pounds_{\kil}g$, its components $\p\defi \phi^{\star}\lp\Kkil(\rig,\rig)\rp$, $\bqone\defi\phi^{\star}\lp\Kkil(\rig,\cdot)\rp$,   and the tensor $\Sigmakil\defi \pounds_{\kil}\nabla$. By construction, $\{\H,\gamma,\ellc,\elltwo,\mathbb{Y},\alpha,\p,\bqone\}$ is  
a \textphdata $(\phi,\rig)$-embedded in $(\M,g)$, so in particular it verifies \eqref{ME:calP} with constraint tensor $\bs{\wt{\Rtensor}}$ (by \textit{4}.) 
and $\bcalP^{\ovkil}=\wt\bcalP$ (by \textit{3}.\ and \eqref{pullback:sigmakil}), i.e.\ 
\begin{align}
\label{Me:mathbbY}
\hspace{-0.33cm}
\mathbb{k}\mathbb{Y}_{ab}=&\spc 
\wt{\calP}_{ab}
-\nablao_a\nablao_b\alpha
-2(\nablao_{(a}\alpha)\mathbb{w}_{b)}
-\alpha\Bigg(
\nablao_{(a}  \mathbb{w}_{b)}
+\mathbb{w}_a\mathbb{w}_b
-\frac{1}{2}\Riemo_{(ab)}
+\frac{1}{2}\wt\Rtensor_{ab} 
+\frac{1}{2}\nablao_{(a} \sone_{b)} 
-\frac{1}{2}\sone_a\sone_b\Bigg).
\end{align}
Now, \textit{2}.\ entails $\wt{\bomega}=\bsone-\mathbb{Y}(n,\cdot)\defi \mathbb{w}$, hence $\kappa=n(\alpha)+\alpha\mathbb{k}_n=\mathbb{k}$.
The combination of \eqref{Me:mathbbY} and the fact that $\kappa\neq0$ then yields $\mathbb{Y}=\bY$ (recall \eqref{kappa:Y:prev}), so that $\{\N,\gamma,\ellc,\elltwo,\mathbb{Y}=\bY\}$ is hypersurface data $(\phi,\rig)$-embedded in $(\M,g)$. Contracting \eqref{Me:mathbbY} with $n^b$ and using Lemma \ref{lem:B:times:n} for $\bs{\rho}=\wt\bomega$ and the definition of $\kappa$, one obtains 
\begin{align}
\label{mideq:iff:first:new:2}
0=&\spc 
\wt{\calP}_{ab}n^b
-\nablao_a(n(\alpha))
-\wt\bomega(n)\nablao_{a}\alpha
-\frac{\alpha}{2}\pounds_{n}\wt\omega_a  
-\frac{\alpha}{2} \nablao_{a}\big(\wt\bomega({n})\big)
-\alpha\wt\Rtensor_{ab}n^b
+\frac{\alpha}{2}\wt\Rtensor_{ab}n^b
\end{align}
where we have added and substracted $\frac{\alpha}{2}\wt{\Rtensor}_{ab}n^b$. Substituting the plus term using \eqref{ConstTensror(n,-)} with $\bomega=\wt\bomega$ and $\Q=\wt\bomega(n)$, equation \eqref{mideq:iff:first:new:2} becomes \eqref{Pi(n,dot)=Rtensor}. 
Once \eqref{Pi(n,dot)=Rtensor} is established, \eqref{Pi(n,dot)=Lie:omega} follows directly after rewriting $\wt{\Rtensor}_{ab}n^b$ according to \eqref{ConstTensror(n,-)} and using the definition of $\kappa$. Finally, the hypersurface data $\hypdata$ satisfies the hypotheses of Lemma \ref{lem:identity:lie:constraint} for  
$\bomega=\wt\bomega$.  
So, \eqref{Lie:Pi:general} holds, and since $\phi^{\star} (\textbf{\textup{Ric}}_g) = \wt{\bs{\Rtensor}}$, \eqref{Lie:Pi} is established. 
Thus, 
\textit{1}.-\textit{4}.\ $ \, \Longrightarrow \, $ \eqref{Pi(n,dot)=Rtensor}-\eqref{Lie:Pi}. Observe that the argument also shows that  \textit{1}.-\textit{4}.\ $ \, \Longrightarrow \, \frac{1}{2}\phi^{\star}(\pounds_{\rig}g)=\bY$.

Let us now prove the converse, namely that \eqref{Pi(n,dot)=Rtensor}-\eqref{Lie:Pi} $ \, \Longrightarrow \, $ \textit{1}.-\textit{4}..  
By \cite[Thm.\ 4.2]{mars2024transverseII} we know that, given null metric hypersurface data together with \textit{any} sequence $\{\mathbb{Y}^{(k)}\}_{k\in\mathbb{N}}$ of symmetric $(0,2)$-tensor fields on $\N$, there exists a Lorentzian manifold $(\M,g)$, an embedding $\phi:\N\longhookrightarrow\M$ and an affine-geodesic rigging $\rig$ so that the data is $(\phi,\rig)$-embedded in $(\M,g)$ and $\mathbb{Y}^{(k)}=\frac{1}{2}\phi^{\star}\big(\pounds_{\rig}^{(k)}g\big)$ for all $k\in\mathbb{N}$. We select the first term $\mathbb{Y}^{(1)}$ of the sequence $\{\mathbb{Y}^{(k)}\}_{k\geq1}$  to be given by $\bY$, i.e.\ $\mathbb{Y}^{(1)}=\bY$. By construction, $\{\N,\gamma,\ellc,\elltwo, \mathbb{Y}^{(1)}=\bY\}$ is $(\phi,\rig)$-embedded in $(\M,g)$, so in particular \textit{1}.\ holds. 
Notice that, at this point, the objects $\wt{\bs{\Rtensor}}$, $\wt{\bs{\Pi}}$,  $\wt{\bomega}$ are unrelated to the quantities $\bs{\Rtensor}$, $\bcalP^{\ovkil}$ and $\bomega$ constructed from $\{\gamma,\ellc,\elltwo, \mathbb{Y}^{(1)}=\bY\}$ according to \eqref{defabsRicci}, \eqref{def:calP:ovsigmakil} and \eqref{defY(n,.)andQ}, respectively. In fact, our main task is to demonstrate that the objects with and without tilde actually coincide. 
Since the constraint tensor for embedded hypersurface data satisfies
$\phi^{\star}(\textup{\textbf{Ric}}_g)=\bs{\Rtensor}$, 
we have 
\begin{align}
\label{Pullback:Ric:U=0}\phi^{\star}\big(\textup{\textbf{Ric}}_g\big){}_{ab}\stackbin{\eqref{defabsRicci}}=  \Riemo_{(ab)}- 2 \pounds_{{n}} \Y_{ab}
- 2\Q \Y_{ab}- 2\nablao_{(a}\omega_{b)} -2\omega_{a}\omega_{b}+  3\nablao_{(a}\sone_{b)}+\sone_{a}\sone_{b},
\end{align}
where $\bomega\defi \bsone-\bY(n,\cdot)$ and $\Q\defi-\bY(n,n)$. Now, consider again any extension $\kil$ of $\phi_{\star}\ovkil$ away from $\phi(\N)$, and let $\Sigmakil\defi \pounds_{\kil}\nabla$.  
Define the symmetric $(0,2)$-tensor $\bcalP^{\ovkil}$ according to 
\eqref{calP:usual:expression}, i.e.\
\begin{align}
\label{def:Pi:U=0} \calP^{\ovkil}_{ab}\defi \nablao_a\nablao_b\alpha+2(\nablao_{(a}\alpha)\omega_{b)}+2\alpha\nablao_{(a}\sone_{b)}+n(\alpha)\Y_{ab}-\alpha(\pounds_{n}\bY)_{ab},
\end{align}
so that 
\begin{align}
\label{last}\Sigmakil(e_a,e_b)\stackbin{\eqref{Sigma:general:xprssn}}=&\spc \lp \nablao_a\nablao_b\alpha+2(\nablao_{(a}\alpha)\omega_{b)}+2\alpha\nablao_{(a}\sone_{b)}+n(\alpha)\Y_{ab}-\alpha(\pounds_{n}\bY)_{ab}\rp \phi_{\star}n=\calP^{\ovkil}_{ab}\hspace{0.05cm}\phi_{\star}n.
\end{align}
Let us prove that $\bomega=\wt\bomega$. This will be a consequence of the fact that $\bY$ satisfies \eqref{kappa:Y:prev}. Contracting \eqref{kappa:Y:prev} 
with $n^b$ and using \eqref{Pi(n,dot)=Rtensor} to get rid of $\wt\Rtensor_{ab}n^b$, and Lemma \ref{lem:B:times:n} for $\bs{\rho}=\wt\bomega$,  
one finds
\begin{align*}
\kappa \lp \sone_{a}-\omega_{a}\rp =&\spc 
\frac{1}{2}\wt\calP_{ab}n^b
+\frac{1}{2}\nablao_a\kappa
-\nablao_a(n(\alpha))
+\kappa \lp  \sone_a-\wt\omega_{a}\rp
-\wt\bomega(n)\nablao_{a}\alpha
-\frac{\alpha}{2}\pounds_{n}\wt\omega_a  
-\frac{\alpha}{2} \nablao_{a}\big(\wt\bomega({n})\big)\\
\stackbin{\eqref{Pi(n,dot)=Lie:omega}}=&\spc 
\kappa \lp  \sone_a-\wt\omega_{a}\rp
\end{align*}
where in the last step we used the definition of $\kappa$ to get rid of the term $\nablao_a\kappa$. Since $\kappa\neq0$, we conclude that $\wt\bomega= \bomega\defi \bsone-\bY(n,\cdot)=\bsone-\frac{1}{2}\phi^{\star}(\pounds_{\rig}g)$, hence item \textit{2}.\ holds. Now, the data $\{\N,\gamma,\ellc,\elltwo,\mathbb{Y}^{(1)}=\bY\}$ fulfills Lemma \ref{lem:identity:lie:constraint} for $\wt\bcalP$, $\wt{\bs{\Rtensor}}$ and $\bomega=\wt\bomega$, so in particular \eqref{Lie:Pi:general} holds. This, together with \eqref{Lie:Pi}, $\kappa\neq0$ and the fact that the zero-set of $\alpha$ has empty interior ensures that $\phi^{\star}\big(\textup{\textbf{Ric}}_g\big)=\wt{\bs{\Rtensor}}$, which proves \textit{4}.. Item \textit{3}.\ then follows from combining \eqref{kappa:Y:prev},  \eqref{Pullback:Ric:U=0}, \eqref{def:Pi:U=0} and \eqref{last}. This concludes the proof of \eqref{Pi(n,dot)=Rtensor}-\eqref{Lie:Pi}
 $ \, \Longrightarrow \, $ \textit{1}.-\textit{4}..   

It only remains to show that the pair $(\phi,\rig)$ satisfies $\frac{1}{2}\phi^{\star}(\pounds_{\rig}g)=\bY$ both under conditions \textit{1}.-\textit{4}.\ and \eqref{Pi(n,dot)=Rtensor}-\eqref{Lie:Pi}. We have already shown that \textit{1}.-\textit{4}.\ $ \, \Longrightarrow \, \frac{1}{2}\phi^{\star}(\pounds_{\rig}g)=\bY$, so it only remains to prove that \eqref{Pi(n,dot)=Rtensor}-\eqref{Lie:Pi} $ \, \Longrightarrow \, \frac{1}{2}\phi^{\star}(\pounds_{\rig}g)=\bY$. Suppose that, from \eqref{Pi(n,dot)=Rtensor}-\eqref{Lie:Pi}, one constructs any Lorentzian manifold $(\M,g)$ satisfying \textit{1}.-\textit{4}.\ but with $\frac{1}{2}\phi^{\star}(\pounds_{\rig}g)=\bY'\neq\bY$. Then \textit{1}.-\textit{4}.\  $ \, \Longrightarrow \, \bY=\frac{1}{2}\phi^{\star}(\pounds_{\rig}g)=\bY'\neq \bY $, which is a contradiction. 
\end{proof}
Null hypersurfaces that admit a foliation by cross-sections are particularly interesting in General Relativity. For instance, in the standard characteristic initial value problem one prescribes data on two null hypersurfaces that intersect transversely on a spacelike surface, which forces them to admit cross-sections. In the following theorem, we address the case when $\N$ has product topology $S\times \mathbb{R}$, $S$ being a cross-section. We show that, given totally geodesic null metric data, together with suitable data on $S$ plus 
two scalar functions and a symmetric $(0,2)$-tensor on the hypersurface, the existence of a Lorentzian manifold realizing the prescribed quantities is guaranteed.
\begin{theorem}\label{thm:existence:sections}
Let $\metdata$ be null metric hypersurface data with $\bU=0$. Assume that $\N$ admits 
a cross-section $\iota:S\longhookrightarrow\N$, i.e.\ a codimension-one submanifold which is intersected precisely once by each integral curve of $n$.
Let $z\in\Fcal^{\star}(\N)$, $V\in\Gamma(T\N)$ be gauge parameters, and consider 
\begin{itemize}
    \item[\textup{(I)}] 
    two functions 
    $\kappa,f\in\Fcal(\H)$
    with gauge behaviour
    $\G_{(z,V)}(\kappa)=\kappa$, $\G_{(z,V)}(f) = \frac{1}{z}\lp f - \frac{n(z)}{z}\rp$, and such that $\kappa$ vanishes nowhere on $\N$;
    \item[\textup{(II)}] a function $\ov\alpha\in\Fcal(S)$ and a covector $\bkilone\in\Gamma(T^{\star}S)$, such that 
    $\G_{(z,V)}(\ov\alpha)= z\vert_S \,   \ov\alpha$, $\G_{(z,V)}(\bkilone) = \bkilone - (z^{-1}dz)\vert_S$;
    \item[\textup{(III)}] a symmetric $(0,2)$-tensor field 
    $\mathcal{T}$ in $S$ 
    with gauge behaviour 
    $\G_{(z,V)}(\mathcal{T})= z\vert_S \, \mathcal{T}$;
    \item[\textup{(IV)}] a symmetric, gauge-invariant $(0,2)$-tensor field $\wt{\bs{\Rtensor}}$ in $\N$ 
    such that $\wt{\bs{\Rtensor}}(n,n)=0$.  
\end{itemize}
Then, there exists a unique function $\alpha\in\Fcal(\N)$, one-form $\wt\bomega\in\Gamma(T^{\star}\N)$ and symmetric $(0,2)$-tensor $\wt\bcalP$ on $\N$ satisfying \eqref{Pi(n,dot)=Rtensor}-\eqref{Lie:Pi} and  
\begin{align}
\wt\bomega(n)=f,\qquad 
\iota^{\star}\alpha=\ov\alpha, \qquad 
\iota^{\star}\wt\bomega=\bkilone,\qquad  
\iota^{\star}\wt\bcalP=\mathcal{T},\qquad  
\kappa=n(\alpha)+\alpha \wt\bomega(n).
\end{align}
Moreover, 
$\G_{(z,V)}(\alpha)=z\alpha$, $\G_{(z,V)}(\wt\bomega)=\wt\bomega-z^{-1}dz$, and $\G_{(z,V)}(\wt\bcalP)=z\wt\bcalP$. Therefore, by Theorem \ref{thm:exist:pi:R} there exists a Lorentzian manifold $(\M,g)$ and an embedding pair $(\phi,\rig)$, such that conditions \textit{1}.-\textit{4}.\ hold, and $\frac{1}{2}\phi^{\star}\big(\pounds_{\rig}g\big)=\bY$, with $\bY$ given by \eqref{kappa:Y:prev}.
\end{theorem}
\begin{remark}
Since $\ov\alpha,\bkilone,\mathcal{T}$ are only defined on $S$, their gauge behaviours are only imposed therein.
\end{remark}
\begin{proof}
Consider a basis $\{v_A\}$ of $\Gamma(TS)$, and extend it to $\N$ by solving $\pounds_nv_A=0$. Then $\{n,v_A\}$ form a basis of $\Gamma(T\N)$ because  $\det\big(\gamma(v_A,v_B)\big)\vert_S\neq0$ and $\pounds_n\big(\gamma(v_A,v_B)\big)=(\pounds_n\gamma)(v_A,v_B)=2\bU(v_A,v_B)=0$. 
We now construct a function $\alpha$, a one-form $\wt\bomega$ and a symmetric $(0,2)$-tensor $\wt\bcalP$ such that conditions \eqref{Pi(n,dot)=Rtensor}-\eqref{Lie:Pi} of Theorem \ref{thm:exist:pi:R} hold:
\begin{itemize}
    \item[1.] We let $\alpha\in\Fcal(\H)$ be the only solution to the ODE $n(\alpha)=\kappa-\alpha f$ with initial data $\alpha \, \stackbin{S}= \, \ov\alpha$. Observe that the set $\{p\in\N \, \colon \, \alpha(p)=0\}$ has empty interior on $\N$ because  $n(\alpha)\vert_p=\kappa\vert_p\neq0$ at any point $p\in\N$ where $\alpha(p)=0$. 
    \item[2.] To construct the one-form $\wt\bomega\in\Gamma(T^{\star}\H)$, we set its component along $n$ to be given by $\wt\bomega(n):= f$, and the remaining components to be the unique solutions of the system of ODEs 
    \begin{align}   
    \label{constr:omega(vA)}\pounds_n\big(\wt\bomega(v_A)\big)=\wt{\bs{\Rtensor}}(n,v_A)+v_A(f), 
    \qquad \text{with} \qquad \wt\bomega(v_A) \, \stackbin{S}= \, 
    \bkilone(v_A).
    \end{align}
    Note that $\kappa=n(\alpha)+\alpha\wt\bomega(n)$ by construction.
    \item[3.] For  
    the tensor $\wt\bcalP$, we proceed as follows: 
    \begin{itemize}
        \item[$(a)$] We set  
        $\wt\bcalP(n,n)\defi  n(\kappa)$.
        \item[$(b)$] We let the 
        components $\wt\bcalP(n,v_A)$ be the unique solutions of the system of ODEs
        \begin{align}
        \nn 0=
        \Big(&\pounds_n\big(\wt\calP_{ab}n^b\big)
        -\nablao_a\big(\wt\bcalP(n,n)\big)
        +2\wt\bcalP(n,n)\big(\sone_a-\wt\omega_{a}-\Y_{ab}n^b\big)\\
        \label{constr:Pi(n,cdot)} &
        -\alpha\pounds_n\big(\wt{\Rtensor}_{ab}n^b\big)
        -n(\alpha)\wt{\Rtensor}_{ab}n^b\Big)v_A^a, 
        \end{align}
        with initial data $\wt\bcalP(n,v_A) \, \stackbin{S}= \, \ov\alpha \,  \wt{\bs{\Rtensor}}(n,v_A)+v_A(\kappa)$, and where $\bY$ is defined by \eqref{kappa:Y:prev}. Observe that \eqref{constr:Pi(n,cdot)} involves metric hypersurface data, the quantities $\alpha$, $\wt\bomega$, $\wt\bcalP(n,n)$ and $\wt{\bs{\Rtensor}}$ which are already defined on $\N$, and the components $\wt\bcalP(n,v_A)$ (note that, besides in the term $\pounds_n\big(\wt\bcalP(n,v_A)\big)$, they also 
        appear algebraically inside $\bY(n,\cdot)$). The key point is that no terms of the form  $\wt\bcalP(v_A,v_B)$ are present, hence \eqref{constr:Pi(n,cdot)} yields a unique solution for the components $\wt\bcalP(n,v_A)$. 
        \item[$(c)$] We let the  
        components $\wt\bcalP(v_A,v_B)$ be the unique solutions of the system of ODEs
    \end{itemize}
        \begin{align}
        \label{constr:piAB} \hspace{-0.9cm}0=\bigg( \big(\pounds_n  +\wt\bomega(n) \big) \wt\calP_{ab}-\nablao_{(a}\big(\wt\calP_{b)c}n^c\big)-\wt\bcalP(n,n)\Y_{ab}-2\wt\omega_{(a}\wt\calP_{b)c}n^c
        -\frac{\alpha}{2}\pounds_n\wt{\Rtensor}_{ab}-(\nablao_{(a}\alpha)\wt{\Rtensor}_{b)c}n^c\bigg) v_A^av_B^b,
        \end{align}
    \begin{itemize}
    \item[] with initial data $\wt\bcalP(v_A,v_B) \, \stackbin{S}= \, \mathcal{T}(v_A,v_B)$, and where again $\bY$ is given by \eqref{kappa:Y:prev}.
    \end{itemize}
\end{itemize}
Now that  $\alpha,\wt\bomega,\wt\bcalP,\wt{\bs{\Rtensor}}$ are defined everywhere on $\N$, let us show that (I)-(IV) $\, \Longrightarrow \, $ \eqref{Pi(n,dot)=Rtensor}-\eqref{Lie:Pi}. 
First, notice that \eqref{constr:Pi(n,cdot)} can be rewritten as 
    \begin{align}
    \label{constr:Pi(n,cdot):rewritten} 0=
    \Big(&\pounds_n\big(\wt\calP_{ab}n^b-\alpha\wt{\Rtensor}_{ab}n^b\big)
    -\nablao_a\big(\wt\bcalP(n,n)\big)
    +2\wt\bcalP(n,n)\big(\sone_a-\wt\omega_{a}-\Y_{ab}n^b\big)\Big)v_A^a.
    \end{align}
On the other hand, contracting 
\eqref{kappa:Y:prev} with $n^b$, applying Lemma \ref{lem:B:times:n}
with $\bs{\rho}=\wt\bomega$, and using 
$n(\alpha)+\alpha\wt\bomega(n)=\kappa$ gives
    \begin{align}
    \label{last:Y(n,.)} -\kappa\big(\sone_a-\wt\omega_a-\Y_{ab}n^b\big)=\wt\calP_{ab}n^b-\frac{\alpha}{2}\pounds_n\wt\omega_a-\nablao_a\big( n(\alpha)\big) -\frac{\alpha}{2}\nablao_a\big(\wt\bomega(n)\big)-\wt\bomega(n)\nablao_a\alpha-\frac{\alpha}{2}\wt{\Rtensor}_{ab}n^b.
    \end{align}
Inserting \eqref{last:Y(n,.)} into \eqref{constr:Pi(n,cdot):rewritten} and using \eqref{constr:omega(vA)}, $(a)$ and the fact that $v_A\big(n(\kappa)\big)=n\big(v_A(\kappa)\big)$, one obtains 
    \begin{align}
    \nn 0=&\spc 
    \pounds_n\Big(\wt\bcalP(n,v_A)-\alpha\wt{\bs{\Rtensor}}(n,v_A)-v_A(\kappa)\Big)
    -\frac{2n(\kappa)}{\kappa}\Big(\wt\bcalP(n,v_A)-\alpha\wt{\bs{\Rtensor}}(n,v_A)-v_A(\kappa)\Big).
    \end{align}
This is a system of ODEs for $\wt\bcalP(n,v_A)-\alpha\wt{\bs{\Rtensor}}(n,v_A)-v_A(\kappa)$ along the null generators of $\H$ with vanishing initial data on $S$. Hence 
$\wt\bcalP(n,v_A)-\alpha\wt{\bs{\Rtensor}}(n,v_A)-v_A(\kappa)=0$ is satisfied everywhere on $\N$. This proves that the contraction of \eqref{Pi(n,dot)=Rtensor} with $v_A$ holds. 
For \eqref{Pi(n,dot)=Rtensor} to hold it only remains to verify that its contraction with $n$ holds true, namely that
$\wt\bcalP(n,n)=\alpha \wt{\bs{\Rtensor}}(n,n)+n(\kappa)$. This is automatically satisfied because $\wt{\bs{\Rtensor}}(n,n)=0$ and we have enforced $\wt\bcalP(n,n)=n(\kappa)$. 

To demonstrate \eqref{Pi(n,dot)=Lie:omega} we combine \eqref{Pi(n,dot)=Rtensor}, \eqref{constr:omega(vA)} and $\kappa=n(\alpha)+\alpha\wt\bomega(n)$ to get 
\begin{align*}
0=&\spc \wt\bcalP(n,v_A)
-\alpha\wt{\bs{\Rtensor}}(n,v_A)
-v_A\big(\kappa\big)= \wt\bcalP(n,v_A)
-\alpha\Big(\pounds_{n}\big(\wt\bomega(v_A)\big)-v_A\big(\wt\bomega(n)\big)\Big)
-v_A\big(n(\alpha)+\alpha\wt\bomega(n)\big)\\
=&\spc\wt\bcalP(n,v_A)
-\alpha\pounds_{n}\big(\wt\bomega(v_A)\big)
-v_A(\alpha)\wt\bomega(n)
-v_A\big(n(\alpha)\big).
\end{align*}
This establishes the validity of  \eqref{Pi(n,dot)=Lie:omega} contracted with $v_A$. 
Its contraction with $n$ is 
$\wt\bcalP(n,n)=
\alpha n\big(\wt\bomega(n)\big)+\wt\bomega(n)n(\alpha)+n\big(n(\alpha)\big)
=n\big(\kappa\big)$, which is fulfilled by construction. Hence, \eqref{Pi(n,dot)=Lie:omega} also holds. 

Finally, to show that \eqref{Lie:Pi} is satisfied, we first note that \eqref{Lie:Pi} can be rewritten  as 
    \begin{align}
    0=\big( \pounds_n +\wt\bomega(n)\big) \wt\calP_{ab}-\nablao_{(a}\big(\wt\calP_{b)c}n^c\big)-\wt\bcalP(n,n)\Y_{ab}
    -2\wt\omega_{(a}\wt\calP_{b)c}n^c
    -\frac{\alpha}{2}\pounds_n\wt{\Rtensor}_{ab}-(\nablao_{(a}\alpha)\wt{\Rtensor}_{b)c}n^c \label{last:proof}
    \end{align}
by using that the set $\{p\in\H \, : \, \alpha(p)=0\}$ has empty interior on $\N$,  together with \eqref{lie:eta:lie:n} for $T=\wt\bcalP,\wt{\bs{\Rtensor}}$. 
To establish \eqref{last:proof} we first note that its contraction with $v^a_A v^b_B$ is just \eqref{constr:piAB}, so it is verified by construction.
Its contraction with $v_A^an^b$ takes the form 
\eqref{constr:Pi(n,cdot)} after using \eqref{n:nablao:theta:sym} for $\bs{\theta}=\wt\bcalP(n,\cdot)$ and $\wt{\bs{\Rtensor}}(n,n)=0$, so it is also verified.  
Finally, the contraction of \eqref{last:proof} with $n^an^b$ is just  $0=\wt\bcalP(n,n)\big(\wt\bomega(n)+\bY(n,n)\big)$,
but $\wt\bomega(n)+\bY(n,n)=0$ (as follows from multiplying \eqref{last:Y(n,.)} by $n^a$ and using $\kappa=n(\alpha)+\alpha\wt\bomega(n)$
and $\wt\bcalP(n,n)=n(\kappa)$). 
So, all contractions of \eqref{last:proof} are satisfied, 
and \eqref{Lie:Pi} holds. By how $\alpha,\wt\bomega,\wt\bcalP$ have been built, it is immediate that they satisfy $\wt\bomega(n)=f$, $\iota^{\star}\alpha=\ov\alpha$, $\iota^{\star}\wt\bomega=\bkilone$ and $\iota^{\star}\wt\bcalP=\mathcal{T}$. 

Let us prove that  
$\alpha, \wt\bomega, \wt\bcalP$ have the gauge behaviours claimed in the theorem (and which are in accordance with the hypotheses
$(ii)$-$(iv)$ in Theorem \ref{thm:exist:pi:R}). 
We consider gauge parameters $z\in\Fcal^{\star}(\N),V\in\Gamma(T\N)$ and denote $\G_{(z,V)}$-transformed quantities with a prime. 
\begin{itemize}
    \item[1.] 
    Using \eqref{gaugen} and the gauge transformation of $\kappa$ and $f$, it is immediate to get 
    \begin{align*}
    n(\alpha)+\alpha f=
    \kappa=\kappa'=
    n'(\alpha')+\alpha'f'
    =n\lp \frac{\alpha'}{z}\rp  
    +\frac{\alpha'}{z} f 
    \quad \Longrightarrow \quad 
    0= n\lp \frac{\alpha'}{z}-\alpha\rp  
    +\lp \frac{\alpha'}{z} -\alpha\rp f.
    \end{align*}
    This is an ODE for  
    $z^{-1}\alpha'-\alpha$ along the  
    generators of $\N$, with vanishing initial data on $S$. Therefore,  
    $\alpha'=z\alpha$ holds everywhere on $\N$. 
    \item[2.] To obtain the gauge-behaviour of $\wt\bomega$, we first notice that it satisfies 
    \begin{align}
    \label{bomega:auxiliary:eq}    
    \pounds_n\wt\bomega=\wt{\bs{\Rtensor}}(n,\cdot)+df.
    \end{align}
    Indeed, 
    its contraction 
    with $n$ is automatically satisfied because $\wt{\bs{\Rtensor}}(n,n)=0$, and its contraction with 
    $v_A$ is just 
    \eqref{constr:omega(vA)}.
    The gauge-transformed $\wt\bomega'$ satisfies \eqref{bomega:auxiliary:eq}   with primes everywhere.
    Combining  \eqref{bomega:auxiliary:eq} and the gauge behaviour of $f$ and $\wt{\bs{\Rtensor}}$, one finds   
    \begin{align*}
    \pounds_{n}\lp\wt\bomega'-\lp \wt\bomega-\frac{dz}{z}\rp\rp\stackbin{ \eqref{gaugen}}=&\spc  z\pounds_{n'}\wt\bomega'+\wt\bomega'(n')dz
    -\pounds_{n} \wt\bomega
    +\frac{1}{z} d\big( n(z)\big)
    -\dfrac{n(z)}{z^2} dz\\
    \stackbin[\eqref{bomega:auxiliary:eq}']{\eqref{bomega:auxiliary:eq}}=
    &\spc  
    zdf' +f'dz-df
    +\frac{1}{z} d\big( n(z)\big)
    -\dfrac{n(z)}{z^2} dz
    =0.
    \end{align*}
    Again, since $\wt\bomega'-\lp \wt\bomega-z^{-1}dz\rp=0$ on $S$, 
    it follows that $\wt\bomega'=\wt\bomega-z^{-1}dz$ everywhere  on $\N$. 
    \item[3.] Finally, to derive the transformation law of  $\wt\bcalP$, we first note that 
    \begin{align}
    z^{-1}\wt\bcalP'(n,\cdot)
    \stackbin{\eqref{gaugen}}=
    \wt\bcalP'(n',\cdot)
    \stackbin{\eqref{Pi(n,dot)=Rtensor}'}=
    \alpha'\wt{\bs\Rtensor}'(n',\cdot)+d\kappa'
    =\alpha\wt{\bs\Rtensor}(n,\cdot)+d\kappa
    \stackbin{\eqref{Pi(n,dot)=Rtensor}}=\wt\bcalP(n,\cdot),
    \end{align}
    so $\wt\bcalP'(n,\cdot)=z\wt\bcalP(n,\cdot)$. Observe that this  
    states 
    that $\wt\bcalP(n,\cdot)$ is gauge-invariant, 
    which in turn is compatible with having imposed $n(\kappa)=\wt\bcalP(n,n)$.
    To obtain the gauge behaviour of the components $\wt\bcalP(v_A ,v_B)$, we 
    compute the $\G_{(z,V)}$-transformation of \eqref{last:proof}. This requires the following two gauge-behaviours:
    \begin{align}
    \label{gauge:nablao:pi:n}\lp\nablao_{(a}\big( \wt\calP_{b)c}n^c\big)\rp'&=
    \nablao_{(a}\big( \wt\calP_{b)c}n^c\big)-\frac{\wt\bcalP(n,n)}{2z}\lp \pounds_{zV}\gamma_{ab}+2\ell_{(a}\nablao_{b)}z\rp,\\
    \label{gauge:kappa:Y:-:pi}\lp\kappa \Y_{ab}-\wt\calP_{ab}\rp'&=\kappa\lp z\Y_{ab}+\ell_{(a}\nablao_{b)}z+\frac{1}{2}\pounds_{zV}\gamma_{ab}\rp
    -z\wt\calP_{ab}.
    \end{align}
    The first one uses \eqref{gauge:nablao} and the already mentioned gauge-invariance of $\wt\bcalP(n,\cdot)$. 
    The second one requires a longer computation. An important step in the derivation is the property that $\Riemo_{(ab)}-\nablao_{(a}\sone_{b)}+s_as_b$ is gauge invariant when $\bU=0$. This was already established under $\G_{(z,0)}$-transformations in \cite[Cor.\ 2.5]{mars2024transverseII}. For $\G_{(1,V)}$-transformations, the result is a consequence of \eqref{gauge:sone}, \eqref{gauge:nablao} and $\nablao_{a}\gamma_{bc}\stackbin{\eqref{nablaogamma}}=0$, together with the general relation  $\Riemo{}'_{ab}= \Riemo{}_{ab}+2 \nablao_{[c} C^c{ }_{b]a}+2 C^c{ }_{d[c} C^d{ }_{b]a}$ 
    between the Ricci tensors of two connections that differ by a tensor $C$. In the present case $C\defi \nablao{}'-\nablao=z^{-1}   n\otimes \big( \frac{1}{2}\pounds_{zV}\gamma+\ellc\otimes_s dz\big)$ by \eqref{gauge:nablao}. With the gauge invariance of $\Riemo_{(ab)}-\nablao_{(a}\sone_{b)}+s_as_b$ at hand, expression \eqref{gauge:kappa:Y:-:pi} follows from \eqref{kappa:Y:prev} after using \eqref{gauge:nablao}, the already established gauge behaviours $\alpha' =  z \alpha$ and
    $\wt\bomega' = \wt\bomega - z^{-1} dz$ as well as the gauge invariance of $\wt{\bs{\Rtensor}}$ and $\kappa$. We now compute
    \begin{align*}
    \frac{1}{z}\pounds_{n}\wt\calP'_{ab}\stackbin{\eqref{gaugen}}=&\spc \pounds_{n'}\wt\calP'_{ab}+\frac{2\nablao_{(a}z}{z}\wt\calP_{b)c}n^c\stackbin{\eqref{last:proof}'}=
    -\wt\bomega'(n')\wt\calP'_{ab}
    +\nablao{}'_{(a}\big(\wt\calP'_{b)c}n'^c\big)
    +\wt\bcalP'(n',n')\Y'_{ab}\\
    &+2\wt\omega'_{(a}\wt\calP'_{b)c}n'^c
    +\frac{\alpha'}{2}\pounds_{n'}\wt{\Rtensor}'_{ab}
    +(\nablao{}'_{(a}\alpha')\wt{\Rtensor}'_{b)c}n'^c 
    +\frac{2\nablao_{(a}z}{z}\wt\calP_{b)c}n^c\\
    \stackbin[\eqref{gauge:kappa:Y:-:pi}]{\eqref{gauge:nablao:pi:n}}=&\spc 
    -\frac{1}{z}\lp\wt\bomega(n)-\frac{n(z)}{z}\rp \wt\calP'_{ab}
    +\nablao_{(a}\big(\wt\calP_{b)c}n^c\big)
    +\wt\bcalP(n,n)\lp  \Y_{ab}+\frac{1}{\kappa z}\lp\wt\calP'_{ab}-z\wt\calP_{ab}\rp\rp
    \\
    &+2\wt\omega_{(a}\wt\calP_{b)c}n^c
    +\frac{\alpha}{2}\pounds_{n}\wt{\Rtensor}_{ab}
    +(\nablao_{(a}\alpha)\wt{\Rtensor}_{b)c}n^c\\
    \stackbin{\eqref{last:proof}}=&\spc -\frac{1}{z}\lp\wt\bomega(n)-\frac{n(z)}{z}\rp \wt\calP'_{ab}
    +\big(\pounds_n+\wt\bomega(n)\big) \wt\calP_{ab}
    +\frac{\wt\bcalP(n,n)}{\kappa z}\lp \wt\calP'_{ab}-z\wt\calP_{ab}\rp\\
    \Longrightarrow\quad 0=&\spc 
    \lp \pounds_n+\wt\bomega(n)-\frac{n(z)}{z}-\frac{\wt\bcalP(n,n)}{\kappa }\rp\lp \wt\calP'_{ab}-z\wt\calP_{ab}\rp.
    \end{align*}
    This is a system of ODEs for all components of $\wt\bcalP'-z\wt\bcalP$ along the null generators, with vanishing initial data on $S$. Therefore, $\G_{(z,V)}\big(\wt\bcalP\big)=z\wt\bcalP$ everywhere on $\N$.
\end{itemize}
We conclude that the data $\{\N,\gamma,\ellc,\elltwo,\alpha,\wt\bomega,\wt\bcalP,\wt{\bs{\Rtensor}}\}$ satisfies all the hypotheses of Theorem \ref{thm:exist:pi:R}, and the last statement in the theorem is justified. 
\end{proof}
We conclude the paper with some comments on Theorems \ref{thm:exist:pi:R} and \ref{thm:existence:sections}.
\begin{remark}\label{rem:1}
The only restriction imposed on the tensor $\wt{\bs{\Rtensor}}$ prescribed in Theorems \ref{thm:exist:pi:R} and \ref{thm:existence:sections} is $\wt{\bs{\Rtensor}}(n,n)=0$. This, as mentioned in Remark \ref{rem:Raychaudhuri}, is necessary for $\wt{\bs{\Rtensor}}$ to be consistent with the Raychaudhuri equation \eqref{ConstTensror(n,n)} with $\bU=0$. Except for this, no other condition is imposed, such as e.g.\ $\wt{\bs{\Rtensor}}$ being zero o proportional to $\gamma$. 
Theorems \ref{thm:exist:pi:R} and \ref{thm:existence:sections} 
are hence independent of any field equations. 
\end{remark}
\begin{remark}\label{rem:2}
In \textup{\cite{mars2024transverseI,mars2024transverseII}}, the authors introduce the notion of abstract Killing horizon data. This consists of null metric hypersurface data $\metdata$ with $\bU=0$, a smooth function $\alpha\in\Fcal(\N)$ whose zero set has empty interior on $\N$, and a one-form $\bs{\tau}\in\Gamma(T^{\star}\N)$ satisfying the conditions
\begin{align*}
(a)\quad \alpha\pounds_n\bs{\tau}=n(\alpha)\bs{\tau}-\bs{\tau}(n)d\alpha,
\qquad \text{and} \qquad 
(b)\quad \alpha^{-1}(\bs{\tau}-d\alpha)\text{ extends smoothly to all }\N.
\end{align*}
The link between the one-form $\bs{\tau}$  
and $\bomega\defi \sone-\bs{r}$ 
is given by $\bs{\tau}=d\alpha+\alpha\bomega$, so condition $(b)$ is automatically satisfied in our case. On the other hand, condition $(a)$ gets rewritten in terms of $\bomega$ as 
\begin{align*}
0=\pounds_n\bomega+\bomega(n)d\alpha+d\big(n(\alpha)\big).
\end{align*}
This corresponds to \eqref{Pi(n,dot)=Lie:omega} with $\wt\bcalP=0$, which is consistent with the fact that Killing horizons have vanishing \isotensor (cf.\ Remark \ref{rem:horizons}). The problem of existence of a $\Lambda$-vacuum spacetime where a given abstract Killing horizon data can be embedded was analyzed in \textup{\cite[Thm.\ 5.8]{mars2024transverseII}}, 
where it was found that it is sufficient to complement $(a)$-$(b)$ with the condition that $\bs{\tau}(n)$ is constant and non-zero. Since $\bs{\tau}(n)$ is precisely $\kappa$, this constancy of $\kappa$ is just \eqref{Pi(n,dot)=Rtensor}. No trace of condition \eqref{Lie:Pi} appeared in \textup{\cite{mars2024transverseII}}, as in the $\Lambda$-vacuum case 
$\wt{\bs{\Rtensor}} = \Lambda \gamma$ and hence \eqref{Lie:Pi} is automatically satisfied whenever $\wt{\bcalP} =0$. 

All in all,  the results in this section can be viewed as a generalization from the Killing horizon case in $\Lambda$-vacuum to the case of arbitrary vector fields $\bar{\eta}$ with no restrictions on the field equations. For such general setting we have demonstrated that \eqref{Pi(n,dot)=Rtensor} and \eqref{Pi(n,dot)=Lie:omega} are not sufficient by themselves and they must be supplemented with \eqref{Lie:Pi} in order to guarantee the existence of a spacetime where the data can be embedded.
\end{remark}
\begin{remark}\label{rem:3}
When the null metric hypersurface data $\metdata$ admits a cross-section $S$, the function $f$ in Theorem \ref{thm:existence:sections} can be chosen to be zero without loss of generality. Indeed, for any choice of $f$, the ODE $n(\ln(z))\stackrel{\N}{=}f$ admits a unique solution $z$ for given initial data $z\vert_S\neq0$ on $S$ (and $z$ must be no-where zero, because if $z(p)=0$ at some $p\in\N$ then $n(z)\vert_p=0$, and $z$ would vanish everywhere along the null generator $\sigma_p$ containing $p$, in particular at $\sigma_p\cap S$). One can therefore take $z$, $V\in\Gamma(T\N)$ as gauge parameters, and construct a Lorentzian manifold $(\mathcal{M},g)$ according to Theorem \ref{thm:existence:sections}, with $f'\defi \G_{(z,V)}(f)=0$. Note that $n(\alpha)=\kappa\neq0$ in this case, which implies that $\alpha$ vanishes at most at one point along each null generator of $\N$.  
\end{remark}
\begin{remark}\label{rem:4}
Theorem \ref{thm:existence:sections} naturally raises the question of under which conditions the constructed Lorentzian manifold $(\M,g)$ admits a horizon, and if so, what are its defining properties. To address this issue, let us consider null metric  
data $\metdata$ with $\bU=0$ and admitting a cross-section $S$,  and 
prescribe $\kappa$, $f$, $\ov{\alpha}$, $\bkilone$ and $\mathcal{T}=0$\footnote{Note that this requirement is gauge-invariant, as $\G_{(z,V)}(\mathcal{T})\stackbin{S}=z\mathcal{T}$. }
satisfying conditions \textup{(I)-(III)} of Theorem \ref{thm:existence:sections}, with the additional restriction that $n(\kappa)=0$ on $\N$. We can then build the function $\alpha$ as in the proof of Theorem \ref{thm:existence:sections} (note that $\wt{\bs{\Rtensor}}$ plays no role in this construction), 
and define the vector field $\ovkil \defi \alpha n$ on $\N$. 
Our aim now is to first find a suitable tensor $\wt{\bs{\Rtensor}}$ fulfilling condition \textup{(IV)} of Theorem \ref{thm:existence:sections}, and then construct the manifold 
$(\M,g)$  and an extension $\kil\in\Gamma(T\M)$ of $\phi_{\star}\ovkil$ so that

\vspace{-0.55cm}

\noindent
\begin{minipage}[t]{0.5\textwidth}
	\begin{align}
		\Kkil\stackbin{\phi(\N)} = 0; \label{eq:(1)}
	\end{align}
\end{minipage}
\hfill
\hfill
\begin{minipage}[t]{0.5\textwidth}
	\begin{align}
		\phi^{\star}(\pounds_{\rig}\Kkil) = 0.  \label{eq:(2)} 
	\end{align}
\end{minipage}
 
We let $\{n,v_A\}$ be a basis of $\Gamma(T\N)$ such that $\pounds_nv_A\stackbin{\N}=0$ and $v_A\vert_S\in\Gamma(TS)$,  
and $\wt{\bs{\Rtensor}}$ be a smooth, symmetric and gauge-invariant $(0,2)$-tensor field on $\N$ satisfying

\vspace{-0.6cm}

\noindent
\begin{minipage}[t]{0.4\textwidth}
\begin{align}
\label{construction:Rtensor:horizon:final:1}
\alpha\wt{\bs{\Rtensor}}(n,\cdot)+d\kappa\stackbin{\N}=0,
\end{align}
\end{minipage}
\hfill
\begin{minipage}[t]{0.6\textwidth}
\begin{align}
\label{construction:Rtensor:horizon:final:2}
\Big(\alpha\pounds_{ n}\wt{\bs{\Rtensor}}+2d\alpha \otimes_s \wt{\bs{\Rtensor}}(n,\cdot)\Big)(v_A,v_B) \stackbin{\N}= 0.
\end{align}
\end{minipage}

\vspace{-0.14cm}

Observe that the contraction of \eqref{construction:Rtensor:horizon:final:1} with $n$ yields $\alpha\wt{\bs{\Rtensor}}(n,n)\stackbin{\N}=0$, and since the zero set of $\alpha$ has empty interior on $\N$ (see the proof of Theorem \ref{thm:existence:sections}),  
$\wt{\bs{\Rtensor}}(n,n)\stackbin{\N}=0$ and 
$\wt{\bs{\Rtensor}}$ is 
in accordance with \textup{(IV)}. Now, we construct $\wt\bcalP(n,\cdot)$ as in the proof of Theorem \ref{thm:existence:sections}. Observe that this tensor satisfies $\wt\bcalP(n,\cdot)=0$ (because $\wt\bcalP(n,n)=n(\kappa)=0$ and $\wt\bcalP(n,v_A)=\alpha \wt{\bs{\Rtensor}}(n,v_A)+v_A(\kappa)=0$). Moreover, inserting $\wt\bcalP(n,\cdot)=0$ and \eqref{construction:Rtensor:horizon:final:2} into  
\eqref{constr:piAB} 
gives
\begin{align*}
0=\big(\pounds_n  +\wt\bomega(n) \big) \wt\bcalP(v_A,v_B),
\end{align*}
and since $\wt\bcalP(v_A,v_B)\vert_S=\mathcal{T}(v_A,v_B)=0$, it follows that $\wt\bcalP=0$ everywhere on $\N$. 

At this point, we can construct a Lorentzian manifold $(\M,g)$ according to Theorem \ref{thm:existence:sections}, and check whether conditions \eqref{eq:(1)}-\eqref{eq:(2)} are fulfilled. We first notice that the data $\metdata$ can always be supplemented with the function $\alpha$ 
(which satisfies $\G_{(z,V)}(\alpha)=z\alpha$, as already proven), a function $\p\in\Fcal(\N)$ and a covector $\bqone\in\Gamma(T^{\star}\N)$ with gauge behaviours \eqref{gauge:bqone:alpha:p:w}. 
Then $\pmdata$ is a \textpmdata and, by Proposition \ref{prop:extension:eta},  
it is $(\phi,\rig)$-embedded in $(\M,g)$ if and only if 
the extension $\kil$  
satisfies \eqref{lie:rig:eta}, in which case $\phi^{\star}\Kkil=2\alpha\bU=0,$ $\phi^{\star}\lp\Kkil(\rig,\cdot)\rp=\bqone$ and $\phi^{\star}\lp\Kkil(\rig,\rig)\rp=\p$ (cf.\ Definition \ref{def:embedded:pdata}). Thus, for \eqref{eq:(1)} to hold it suffices to  
fix $\bqone=0$ and $\p=0$, or equivalently the first transverse derivative of $\kil$ to be given by (cf.\ Proposition \ref{prop:extension:eta})
\begin{align}
\label{lie:rig:eta:horizon}
\pounds_{\rig}\kil\stackbin{\phi(\N)}=-n(\alpha)\rig
-\phi_{\star}\lp 
\frac{\alpha}{2} n(\elltwo) \, n
+P\big(2\alpha\bsone+d\alpha, \, \cdot \, \big)\rp .
\end{align}
Finally, condition \eqref{eq:(2)} is equivalent to $\wt{\bcalP}\stackbin{\N}=0$ (or $\bcalP^{\ovkil}\stackbin{\N}=0$, by item \textit{3}.\ of Theorem \ref{thm:exist:pi:R}). This follows from \eqref{id:pi:and:lie_rig:def:tensor}, since in the present case $\mathcal{W}\vert_{\phi(\N)}=0$ and $\phi^{\star}(\pounds_{\mathcal{W}}g)$ only depends on the values of $\mathcal{W}$ at $\phi(\N)$. 
Both \eqref{eq:(1)} and \eqref{eq:(2)} hold, therefore $\phi(\N)$ is a horizon  
in the sense that $(\M,g)$ admits a vector $\kil$, null and tangent to $\phi(\N)$, which also verifies $(\pounds_{\kil}g)\vert_{\phi(\N)} = 0$ and $\phi^{\star}(\pounds_{\rig}\pounds_{\kil}g) = 0$. Note that the condition $n(\kappa)=0$ that we have imposed in this construction is necessary, otherwise $\wt{\bcalP}(n,n)$ cannot vanish identically (cf.\ proof of Theorem \ref{thm:existence:sections}).
\end{remark}

\section*{Acknowledgements}

Both authors acknowledge funding from the grant PID2024-158938NB-I00 funded by MICIU/AEI/ 10.13039/501100011033 and by ``ERDF A way of making Europe". In addition, M. Mars acknowledges financial support under projects SA097P24 (JCyL) and RED2022-134301-T funded by MCIN/AEI/ 10.13039/501100011033, whereas M. Manzano acknowledges the Austrian Science Fund (FWF) [Grant DOI 10.55776/EFP6].

\begingroup
\let\itshape\upshape
\bibliographystyle{plain}
\bibliography{ref}

\end{document}